\newif\iflong
\longtrue %

\PassOptionsToPackage{dvipsnames,table}{xcolor}
\documentclass[a4paper,USenglish]{lipics-v2021}

\usepackage[vlined,linesnumbered,noresetcount]{algorithm2e}
\usepackage{amsmath}
\usepackage{amsthm}
\usepackage[shortlabels]{enumitem}
\usepackage{graphicx}
\usepackage{hyperref}
\usepackage{makecell}
\usepackage{placeins}
\usepackage{setspace}
\usepackage{subcaption}
\usepackage{todonotes}
\usepackage{varwidth}
\usepackage{xspace}
\usepackage{mathrsfs}

\usepackage{cite}

\usepackage{color}
\usepackage{tikz}
\usetikzlibrary{decorations.pathreplacing,positioning,automata,calc}
\usetikzlibrary{shapes,arrows}
\usepgflibrary{shapes.symbols}
\usetikzlibrary{shapes.symbols}
\usetikzlibrary{patterns}
\usetikzlibrary{decorations.pathreplacing}
\usetikzlibrary{decorations.markings}

\newcommand{\newid}{{\tt new\_id}}
\newcommand{\submit}{{\tt submit}}
\newcommand{\exec}{{\tt execute}}
\newcommand{\recover}{{\tt recover}}

\newcommand{\MPreAccept}{\mathtt{PreAccept}}
\newcommand{\MCommit}{\mathtt{Commit}}

\newcommand{\MPreAcceptOK}{\mathtt{PreAcceptOK}}
\newcommand{\MAccept}{\mathtt{Accept}}
\newcommand{\MAcceptOK}{\mathtt{AcceptOK}}
\newcommand{\MRecover}{\mathtt{Recover}}
\newcommand{\MRecoverOK}{\mathtt{RecoverOK}}

\newcommand{\MTryRecover}{\mathtt{TryRecover}}
\newcommand{\MValidate}{\mathtt{Validate}}
\newcommand{\MValidateOK}{\mathtt{ValidateOK}}
\newcommand{\GST}{\mathsf{GST}}

\newcommand{\cmd}{\mathsf{cmd}}
\newcommand{\dep}{\mathsf{dep}}
\newcommand{\coord}{{\tt initCoord}}
\newcommand{\initDep}{\mathsf{initDep}}%
\newcommand{\initCmd}{\mathsf{initCmd}}
\newcommand{\bal}{\mathsf{bal}}
\newcommand{\cbal}{\mathsf{abal}}
\newcommand{\phase}{\mathsf{phase}}
\newcommand{\startP}{{\normalfont\textsc{initial}}}
\newcommand{\preacceptP}{{\normalfont\textsc{preaccepted}}}
\newcommand{\preacceptnotokP}{{\normalfont\textsc{preaccepted-notok}}}
\newcommand{\acceptP}{{\normalfont\textsc{accepted}}}
\newcommand{\commitP}{{\normalfont\textsc{committed}}}

\newcommand{\Executed}{\mathsf{executed}}

\newcommand{\vcbal}{\mathit{abal}}
\newcommand{\vcmd}{\mathit{cmd}}
\newcommand{\vdep}{\mathit{dep}}
\newcommand{\vinitDep}{\mathit{initDep}}%

\newcommand{\vphase}{\mathit{phase}}

\newcommand{\bmax}{b_{\rm max}}
\newcommand{\Rmax}{R_{\rm max}}
\newcommand{\id}{\ensuremath{\mathit{id}}}
\newcommand{\from}{\textbf{from}\xspace}

\newcommand{\send}{\textbf{send}\xspace}
\newcommand{\To}{\textbf{to}\xspace}
\newcommand{\ToAll}{\textbf{to all}\xspace}
\newcommand{\precond}{\textbf{pre:}\xspace}
\newcommand{\onreceive}{\textbf{when received}\xspace}
\newcommand{\Let}{\textbf{let}\xspace}
\newcommand{\assign}[2]{\ensuremath{#1} \ensuremath{\leftarrow} \ensuremath{#2}}
\newcommand\inconflict[2]{\ensuremath{#1 \bowtie #2}}
\newcommand{\conflicts}{\mathit{I}} %

\newcommand{\epaxos}{\textsc{EPaxos}\xspace}
\newcommand{\epp}{\textsc{EPaxos}*\xspace}

\newcommand{\MWaiting}{\mathtt{Waiting}}

\newcommand{\noop}{\mathsf{Nop}}

\newcommand{\tr}[2]{\iflong{}\S#1\else{}\cite[\S#2]{ext}\fi}
\newcommand{\tra}[2]{\iflong{}(\S#1)\else{}\cite[\S#2]{ext}\fi}

\SetKwBlock{SubAlgoBlock}{}{end}
\newcommand{\SubAlgo}[2]{#1 \SubAlgoBlock{\renewcommand{\;}{\\} #2}}

\theoremstyle{definition}
\newtheorem{invariant}{Invariant}
\theoremstyle{propdefinition}

\let\oldnl\nl%
\newcommand{\nonl}{\renewcommand{\nl}{\let\nl\oldnl}}

\newcommand{\procSet}{\ensuremath{\Pi}}
\newcommand{\propose}{\ensuremath{\mathit{propose}}}

\let\emptyset\varnothing

\tikzstyle{message} = [draw, -latex',Black!60, shorten <=0.2em]
\tikzstyle{messageA} = [draw, -latex',Blue!40, shorten <=0.2em]
\tikzstyle{messageB} = [draw, -latex',Red!40, shorten <=0.2em]
\tikzstyle{lifeline} = [draw, Gray!40]

\renewcommand{\_}{\texttt{\textunderscore}}

\DeclareRobustCommand{\hlc}[1]{#1}

\definecolor{HLcolor}{RGB}{255, 226, 148}
\sethlcolor{HLcolor}

\newcommand{\labsection}[1]{\label{sec:#1}}
\newcommand{\refsection}[1]{\S\ref{sec:#1}}

\newcommand{\labappendix}[1]{\label{appendix:#1}}

\newcommand{\labdef}[1]{\label{def:#1}}

\newcommand{\labinv}[1]{\label{inv:#1}}

\newcommand{\command}[2]{\ensuremath{\color{#1}{#2}}}
\newcommand{\cmdc}{\command{Black}{c}}
\newcommand{\cmdcp}{\command{Black}{c'}}
\newcommand{\cmdcpp}{\command{Black}{c''}}

\newcommand{\nlower}{A}
\newcommand{\nthrifty}{B}
\newcommand{\nliveness}{C}
\newcommand{\nproof}{D}
\newcommand{\nbug}{E}

\iflong
\title{Making Democracy Work: Fixing and Simplifying Egalitarian Paxos (Extended Version)}
\else
\title{Making Democracy Work: Fixing and Simplifying Egalitarian Paxos}
\fi

\nolinenumbers

\author{Fedor Ryabinin}{IMDEA Software Institute, Madrid, Spain}{}{}{}
\author{Alexey Gotsman}{IMDEA Software Institute, Madrid, Spain}{}{}{}
\author{Pierre Sutra}{Institut Polytechnique de Paris, Palaiseau, France}{}{}{}


\authorrunning{F. Ryabinin, A. Gotsman, and P. Sutra}
\titlerunning{Making Democracy Work: Fixing and Simplifying Egalitarian Paxos}

\Copyright{Fedor Ryabinin, Alexey Gotsman, and Pierre Sutra}

\iflong
\else
\relatedversiondetails[cite={ext}]{Extended Version}{http://arxiv.org/abs/2511.02743}
\fi

\supplementdetails[subcategory={TLA$^+$ specification}]{Software}{https://arxiv.org/abs/2511.02743}

\funding{This work was partially supported by the projects BYZANTIUM, funded by
  MCIN/AEI and NextGenerationEU/PRTR; DECO and María de Maeztu, funded by
  MICIU/AEI; and REDONDA, funded by the EU CHIST-ERA network.}

\acknowledgements{We thank Benedict Elliott Smith, whose work on leaderless
  consensus in Apache Cassandra inspired this paper, and 
  Alexandre Siret, who wrote a TLA$^+$ specification of our protocol.}

\EventEditors{Andrei Arusoaie, Emanuel Onica, Michael Spear, and Sara Tucci-Piergiovanni}
\EventNoEds{4}
\EventLongTitle{29th International Conference on Principles of Distributed Systems (OPODIS 2025)}
\EventShortTitle{OPODIS 2025}
\EventAcronym{OPODIS}
\EventYear{2025}
\EventDate{December 3--5, 2025}
\EventLocation{Ia\c{s}i, Romania}
\EventLogo{}
\SeriesVolume{361}
\ArticleNo{22}

\begin{document}

\maketitle

\begin{abstract}
  Classical state-machine replication protocols, such as Paxos, rely on a distinguished leader process to order commands.
  Unfortunately, this approach makes the leader a single point of failure and increases the latency for clients that are not co-located with it.
  As a response to these drawbacks, {\em Egalitarian Paxos}~\cite{epaxos} introduced an alternative, leaderless approach, that allows replicas to order commands collaboratively.
  Not relying on a single leader allows the protocol to maintain non-zero throughput with up to $f$ crashes of {\em any} processes out of a total of $n = 2f+1$.
  The protocol furthermore allows any process to execute a command $c$ fast, in $2$ message delays, provided no more than $e = \lceil\frac{f+1}{2}\rceil$ other processes fail, and all concurrently submitted commands commute with $c$; the latter condition is often satisfied in practical systems.
  
Egalitarian Paxos has served as a foundation for many other replication protocols.
But unfortunately, the protocol is very complex, ambiguously specified and suffers from nontrivial bugs.
In this paper, we present \epp{} -- a simpler and correct variant of Egalitarian Paxos.
Our key technical contribution is a simpler failure-recovery algorithm, which we have rigorously proved correct.
Our protocol also generalizes Egalitarian Paxos to cover the whole spectrum of failure thresholds $f$ and $e$ such that $n \ge \max\{2e+f-1, 2f+1\}$ -- the number of processes that we show to be optimal.

\end{abstract}

\section{Introduction}
\labsection{intro}

{\em State-machine replication (SMR)} is a classic approach to implementing fault-tolerant services~\cite{clocks,smr}.
In this approach, the service is defined by a deterministic state machine, and each process maintains its own local replica of the machine.
An SMR protocol coordinates the execution of client commands at the processes to ensure that the system is {\em linearizable}~\cite{linearizability}, i.e., behaves as if commands are executed sequentially by a single process.
Classic protocols for implementing SMR under partial synchrony, such as Paxos~\cite{paxos} and Raft~\cite{raft}, rely on a distinguished leader process that determines the order of command execution.
Unfortunately, this approach makes the leader a single point of failure and increases the latency for clients that are not co-located with it.

As a response to these drawbacks, {\em Egalitarian Paxos} (\epaxos)~\cite{epaxos} introduced a leaderless approach that allows processes to order commands collaboratively.
This more democratic ordering mechanism brings several benefits, the first of which is faster command execution.
In more detail, as is standard, \epaxos tolerates up to $f$ processes crashing out of a total of $n=2f+1$.
But when there are no more than $\lceil\frac{f+1}{2}\rceil$ crashes and the system is synchronous, the protocol additionally allows any correct process to execute a command $c$ in only $2$ message delays, provided $c$ commutes with all concurrently submitted commands; the latter condition is often satisfied in practice.
In this case the protocol takes a {\em fast path}, which contacts {\em a fast quorum} of $n - \lceil\frac{f+1}{2}\rceil = f+\lfloor\frac{f+1}{2}\rfloor$ processes.
In contrast, in Paxos, a non-leader process requires at least $3$ messages delays to execute a command, as it has to route the command through the leader.
Furthermore, a leader failure prevents Paxos from executing any commands, no matter quickly or slowly, until a new leader is elected.
\epaxos can maintain non-zero throughput with up to $f$ failures of {\em any} processes -- a practically important feature for highly available systems.

Interestingly, in comparison to existing SMR protocols with fast paths~\cite{gpaxos,gbroadcast}, \epaxos can provide fast decisions while tolerating an extra process failure.
For example, Generalized Paxos~\cite{gpaxos} decides fast when the number of failures does not exceed $\lfloor\frac{f}{2}\rfloor <  \lfloor\frac{f}{2}\rfloor + 1 = \lceil\frac{f+1}{2}\rceil$.
Thus, for $n = 3$ and $n=5$, \epaxos provides fast decisions while tolerating any minority of failures, with its fast path accessing the remaining majority -- same as for slow decisions in the usual Paxos.
Providing fast decisions in these cases does not require contacting larger quorums.
The above advantages are no small change in practical systems, especially those deployed over wide-area networks.
In this setting, the number of replicas is typically small (e.g., $n \in \{3,5\}$~\cite{spanner}), while the latencies are high.
Thus, waiting for an extra message or contacting an additional far-away replica may cost hundreds of milliseconds per command.

Due to these considerations, \epaxos has been influential: it has served as a foundation for alternative leaderless SMR protocols~\cite{caesar,tempo,atlas}, protocols for distributed transactions and storage~\cite{gryff,shuai-osdi14,shuai-osdi16}, and performance-evaluation studies~\cite{epaxos-revisited}.
\epaxos has also motivated the study of theoretical frameworks for leaderless consensus~\cite{guerraoui-leaderless} and SMR~\cite{losa16,sutra20}.
More recently, it inspired a new protocol for strongly consistent transactions in Apache Cassandra, a popular NoSQL datastore~\cite{accord}.
Overall, the original SOSP'13 paper on \epaxos~\cite{epaxos} has been cited more than 500 times.
Surprisingly, despite all this impact, \epaxos remains poorly understood and is subject to controversies.
As we explain next, this is because the protocol is complex, ambiguously specified, and suffers from nontrivial bugs.

\subparagraph{Complexity.}
At a high level, \epaxos implements SMR by getting the processes to agree on a {\it dependency graph} -- a directed graph whose vertices are commands and edges represent constraints on their execution order.
Processes then execute commands in an order determined by the graph.
To ensure the agreement on the dependency graph, the process that submits a particular command drives the consensus on the set of predecessors of the command in the graph.
If this initial command {\em coordinator} fails, another process takes over its job.
This failure recovery procedure is the most subtle part of the protocol.

The original \epaxos paper~\cite{epaxos} describes in detail only a simplified version of the protocol, which has fast quorums of size $n-1 = 2f$, but in exchange, has a simple failure recovery procedure.
The paper only sketches the full version of the protocol with fast quorums of size $f+\lfloor \frac{f+1}{2} \rfloor$, where recovery is much more complex.
Unfortunately, more detailed descriptions of this recovery procedure available in technical reports~\cite{epaxos-thesis,epaxos-tr} are ambiguous and incomplete.
The TLA$^+$ specification of the protocol~\cite{epaxos-thesis} does not cover recovery in full, and the open-source implementation~\cite{epaxos-code} is not aligned with the available protocol descriptions.

\subparagraph{Bugs.}
As shown in~\cite{sutra19}, \epaxos has a bug in its use of Paxos-like ballot variables.
The usual Paxos partitions the execution of the protocol into {\em ballots}, which determine the current leader~\cite{paxos}.
Each Paxos process uses two variables storing ballot numbers -- the ballot the process currently participates in, and the one where it last voted.
\epaxos uses a similar concept of ballots for each single command, to determine
which process is responsible for computing this command's position in the execution order.
However, the protocol uses only a single ballot variable, which may lead to non-linearizable executions~\cite{sutra19}.
This bug can be easily fixed for the simple version of \epaxos with fast quorums of size $n-1$.
But as we explain in \S\ref{sec:comparison}, fixing it for the full version of the protocol with fast quorums of size $f+\lfloor\frac{f+1}{2}\rfloor$ is nontrivial, because in this case recovery is much more complex than in Paxos.
We have also found a new bug where the recovery can get deadlocked even during executions with finitely many submitted commands (\S\ref{sec:comparison} and \tr{\ref{sec:bug}}{\nbug}).

\subparagraph{Lack of rigorous proofs.}
The available correctness proofs for the version of \epaxos with smaller fast quorums~\cite[\S{}3.6.5]{epaxos-thesis} are not rigorous: they make arguments without referring to the actual pseudocode of the protocol and often posit claims without sufficient justification.
More rigor may have permitted to find that the protocol is incorrect.

Moreover, the existing proofs are given under an assumption that a process submitting a command selects a fast quorum it will use a priori, and communicates only with this quorum.
This version of the protocol, which the authors call {\em thrifty}, takes the fast path more frequently in failure-free scenarios and is simpler to prove correct.
But it does not allow a process to take the fast path even with a single failure, if this failure happens to be within the selected fast quorum.
The non-thrifty version of \epaxos is more robust.
However, it is underspecified, and the authors of \epaxos never proved its correctness.

\subparagraph{Contributions.}
We are thus left in an uncomfortable situation where a protocol underlying a significant body of research suffers from multiple problems.
In this paper we make the following contributions to improve this situation:
\begin{itemize}
\item
  {\em Clarity.}
  We propose a simpler and correct variant of Egalitarian Paxos, which we call \epp.
  Our protocol is explained pedagogically, starting from key notions and invariants.
  The main technical novelty of \epp is its failure recovery procedure, which significantly simplifies the one in the original \epaxos and allows us to fix its bugs.
  At the same time, \epp does not significantly change the failure-free processing of the original protocol, thus preserving its steady-state performance properties.
\item
{\em Correctness.} We rigorously demonstrate the correctness of \epp, proving key properties as we explain the protocol, and covering the rest of the proof in \tr{\ref{sec:proof}}{\nproof}.
  We give proofs for both non-thrifty and thrifty versions of the protocol.
\item
{\em Optimality.}  We generalize the original \epaxos to handle a broader spectrum of failure thresholds under which it can execute commands fast ($e$) or execute any commands at all ($f$).
  To this end, we formalize the desired fault-tolerance properties by a notion of an {\em $f$-resilient $e$-fast} SMR protocol, and show that \epp satisfies it when $n \ge \max\{2e+f-1, 2f+1\}$.
  Building on our recent result about fast consensus~\cite{BA}, we also show that this number of processes is optimal.
  The original \epaxos aims to handle the special case of this inequality when $f = \lfloor \frac{n-1}{2}\rfloor$ and $e = \lceil\frac{f+1}{2}\rceil$, but falls short due to its bugs.
  Our version is thus the first provably correct SMR protocol with optimal values of $e$ and $f$.
  Given a fixed $n$, it permits trading off lower overall fault tolerance ($f$) for higher fault tolerance of the fast path ($e$).
\end{itemize}

\section{Fast State-Machine Replication}
\labsection{smr}

We consider a set $\Pi$ of $n \ge 3$ processes communicating via reliable links.
At most $f$ of the processes can crash.
The system is partially synchronous~\cite{dls}: after some global stabilization time ($\GST$) messages take at most $\Delta$ units of time to reach their destinations.
The bound $\Delta$ is known to the processes, while $\GST$ is unknown.
Processes are equipped with clocks that can drift unboundedly from real time before $\GST$, but can accurately measure time thereafter.

{\em State-machine replication} (SMR) is a common approach to implementing fault-tolerant services in this system~\cite{smr}.
A service is defined by a deterministic state machine with an appropriate set of commands.
Processes maintain their own local copy of the state machine and accept commands from clients (external to the system).
An {\em SMR protocol} coordinates the execution of commands at the processes, ensuring that service replicas stay in sync.
The protocol provides an operation $\submit(c)$, which allows a process to submit a command $c$ for execution on behalf of a client.
The protocol may also trigger an event $\exec(c)$ at a process, asking it to apply $c$ to the local service replica; after execution, the process that submitted the command may return the outcome of $c$ to the client.
Without loss of generality, we assume that each submitted command is unique.

A fault-tolerant service implemented using SMR is supposed to be {\em linearizable}~\cite{linearizability}.
Informally, this means that commands appear as if executed sequentially on a single copy of the state machine in an order consistent with the real-time order, i.e., the order of non-overlapping command invocations.
As observed in~\cite{gbroadcast, gpaxos} for the replicated service to be linearizable, the SMR protocol does not need to ensure that commands are executed at processes in the exact same order: it is enough to agree on the order of non-commuting commands.
Formally, two commands $c$ and $c'$ {\it commute} if in every state $s$ of the state machine: {\em (i)} executing $c$ followed by $c'$ or $c'$ followed by $c$ in $s$ leads to the same state; and {\em (ii)} $c$ returns the same response in $s$ as in the state obtained by executing $c'$ in $s$, and vice versa.
If $c$ and $c'$ do not commute, we say that they {\em conflict} and write $\inconflict{c}{c'}$.
Then the specification of the SMR protocol sufficient for linearizability is given by the following properties:
\begin{description}
\item[{\sf \textbf{Validity.}}] If a process executes a command $c$, then some process has submitted $c$ before.
\item[{\sf \textbf{Integrity.}}] A process executes each command at most once.
\item[{\sf \textbf{Consistency.}}] Processes execute conflicting commands in the same order.
\item[{\sf \textbf{Liveness.}}] In any execution with finitely many submitted commands, if a command is submitted by a correct process or executed at some process, then it is eventually executed at all correct processes.
\end{description}

Note that Egalitarian Paxos is not guaranteed to terminate when infinitely many commands are submitted~\cite{tempo}, hence the extra assumption in \textsf{Liveness}.

We next define a class of SMR protocols that can execute commands fast under favorable conditions.
One such condition is that the system is synchronous in the following sense.
We say that the events that happen during the time interval $[0, \Delta)$ form {\em the first round}, the events that happen during the time interval $[\Delta, 2\Delta)$ {\em the second round}, and so on.
\begin{definition}
  Given $E \subseteq \procSet$, a run is \textbf{$\boldsymbol{E}$-faulty synchronous}, if:
  \begin{itemize}[noitemsep,topsep=0pt,parsep=0pt]
  \item All processes in $\procSet \setminus E$ are correct, and all processes in $E$ are faulty.
  \item Processes in $E$ crash at the beginning of the first round.
  \item All messages sent during a round are delivered %
    at the beginning of the next round.
  \item All local computations are instantaneous.
  \end{itemize}
\end{definition}

Another condition for a command to be executed fast is that there are no concurrent conflicting commands, as defined in the following.
\begin{definition}
  We say that $c$ \textbf{precedes} $c'$ in a run if $c'$ is submitted at a process when this process has already executed $c$.
  Commands $c$ and $c'$ are \textbf{concurrent} when neither precedes the other.
  A command $c$ is \textbf{conflict-free} if it does not conflict with any concurrent command.
\end{definition}
\begin{definition}
  \labdef{fast}
  A run of an SMR protocol is \textbf{fast} when any conflict-free command submitted at time $t$ by a process $p$ gets executed at $p$ by time $t+2\Delta$.
  A protocol is \textbf{$\boldsymbol{e}$-fast} if for all $E \subseteq \procSet$ of size $e$, every $E$-faulty synchronous run of the protocol is fast.
  A protocol is \textbf{$\boldsymbol{f}$-resilient} if it implements SMR under at most $f$ process failures.
\end{definition}

Thus, $e$ defines the maximal number of failures the protocol can tolerate while executing commands fast, and $f$ the failures it can tolerate while executing any commands at all (we assume $e \le f$).
These thresholds are subject to a trade-off stated by the following theorem, which follows from our recent result about consensus~\cite{BA}. We defer its proof to \tr{\ref{appendix:lowerbound}}{\nlower}.
\begin{theorem}
  \label{theo:lower}
  Any $f$-resilient $e$-fast SMR protocol requires $n \ge \max\{2e+f-1, 2f+1\}$.
\end{theorem}

We compare this lower bound with related ones in \S\ref{sec:related}.
In classic SMR protocols, such as Paxos~\cite{paxos}, commands submitted at the leader execute within $2\Delta$ in any $\lfloor \frac{n-1}{2} \rfloor$-faulty synchronous run, but commands submitted at other processes do not enjoy such a low latency.
Hence, these ``undemocratic'' protocols are not $e$-fast for any $e$.
\epaxos aims to be $\lceil\frac{f+1}{2}\rceil$-fast for $n = 2f+1 = 2e + f - 1$, which would match the bound in Theorem~\ref{theo:lower}.
Unfortunately, as we explained in \S\ref{sec:intro}, the corresponding version of the protocol is not well-specified and suffers from nontrivial bugs.
In the rest of the paper we present \epp{} -- a simpler and correct variant of Egalitarian Paxos.
\epp is the first provably correct SMR protocol that matches the lower bound in Theorem~\ref{theo:lower}.
It furthermore generalizes Egalitarian Paxos to cover the whole spectrum of parameters $e$ and $f$ allowed by the theorem: e.g., for a fixed $n$ it allows decreasing $f$ in exchange for increasing $e$.
For instance, when $n\ge 5$ we can let $e = f = \lfloor \frac{n+1}{3} \rfloor$, which maximizes $e$ to the extent possible.

\section{\epp: Baseline Protocol for {\boldmath $n \geq \max\{2e+f+1, 2f+1\}$}}
\labsection{base}

To simplify the presentation, we explain \epp in two stages.
In this section we present a simpler version of the protocol that requires $n \geq \max\{2e+f+1, 2f+1\}$ -- same as in Generalized Paxos~\cite{gpaxos}.
In the next section we explain how to reduce the number of processes to $n \geq \max\{2e+f-1, 2f+1\}$, so that the protocol matches the lower bound in Theorem~\ref{theo:lower}.

\epp implements SMR by getting the processes to agree on a {\it dependency graph} -- a directed graph whose vertices are application commands and edges represent constraints on their execution order.
We say that a command $c$ is a \emph{dependency} of $c'$ when $(c,c')$ is an edge in the graph.
Processes progressively agree on the same dependency graph by agreeing on the set of dependencies of each vertex.
Processes then execute commands in an order determined by the graph.
The graph may contain cycles, which are broken deterministically as we describe in the following.

In \epp, each submitted command is associated with a unique numerical {\em identifier}.
An array $\cmd$ at each process maps command identifiers to their actual payloads.
Processes also maintain local copies of the dependency graph in an array $\dep$, which maps the identifier of a command to the set of identifiers of its dependencies (a {\em dependency set}).
A command $\id$ goes through several {\em phases} at a given process -- $\startP$, $\preacceptP$, $\acceptP$ and $\commitP$ -- with the $\commitP$ phase indicating that consensus has been reached on the values of $\cmd[\id]$ and $\dep[\id]$.
An array $\phase$ maps command identifiers to their current phases (by default, $\startP$).
A command is \emph{pre-accepted}, \emph{accepted}, or \emph{committed} when its identifier is in the corresponding phase.

In case of failures, the protocol may end up committing a command $\id$ with a special $\noop$ payload that is not executed by the protocol and conflicts with all commands (so in the end $\cmd[\id] = \noop$ at all processes).
In this case, the command should be resubmitted (for simplicity, we omit the corresponding mechanism from the protocol description).
The protocol ensures the following key invariants:
\begin{invariant}[Agreement]
  \label{inv:agr}
  If a command $\id$ is committed at two processes with dependency sets $D$ and $D'$ and payloads $c$ and $c'$ then $D = D'$ and $c = c'$.
\end{invariant}
\begin{invariant}[Visibility]
  \label{inv:cnst}
  If distinct commands $\id$ and $\id'$ are both committed with conflicting non-$\noop$ payloads and dependency sets $D$ and $D'$, respectively, then either $\id \in D'$ or $\id' \in D$.
\end{invariant}
These invariants are enough to ensure that conflicting commands are executed in the same order at all processes, as we describe next.

\subsection{Execution Protocol}
\label{sec:execution}

\begin{figure}
    \!\!\!\!\begin{minipage}{\linewidth}
\scalebox{0.9}{
      \begin{algorithm}[H]
        \renewcommand{\;}{\\}
        \SetKwProg{Function}{}{:}{}
        \DontPrintSemicolon
        {\bf while true} \label{epaxos:exec:loop} \SubAlgoBlock{
          \mbox{\Let $G \subseteq \dep$ be the largest subgraph such that $\forall \id \in G.\, \phase[\id] = \commitP \land \dep[\id] \subseteq G$}\;
          \label{epaxos:exec:let-c}
          \For{$C \in \mathsf{SCC}(G)$ {\rm in topological order}}{
            \label{epaxos:exec:for-scc}
            \For{$\id \in C$ {\rm in the order of command identifiers}}{
              \label{epaxos:exec:min}
              \If{$\id \notin \Executed \land \cmd[\id] \neq \noop$}{
                \label{epaxos:exec:test}
                $\exec(\cmd[\id])$\;
                \label{epaxos:exec:exec}
                $\assign{\Executed}{\Executed\cup\{\id\}}$\;
              \label{epaxos:exec:set-exec}
              }
            }
          }
        }
      \end{algorithm}
    }
  \end{minipage}
  \caption{\epp: execution protocol.}
  \label{epaxos:exec}
\end{figure}

Each process runs a background task that executes commands after they are committed (Figure~\ref{epaxos:exec}).
At a high level, the execution protocol works as follows.
A command is ready to be executed once it is committed and so are all its transitive dependencies.
When this happens, the protocol executes all ready commands in a batch in an order consistent with the dependency graph, breaking cycles using command identifiers.
The identifiers of the executed commands are added to an $\Executed$ variable.
More concretely, the execution protocol performs the following steps.
\begin{enumerate}
\item
  \label{exec:one}
  Compute the largest subgraph $G$ of the dependency graph consisting of committed commands that transitively depend only on other committed commands (line~\ref{epaxos:exec:let-c}).
\item
  \label{exec:two}
  Split $G$ into strongly connected components and sort them in the topological order (line~\ref{epaxos:exec:for-scc}).
\item
  \label{exec:three}
  Execute the commands in each component in the identifier order, skipping $\noop$s (lines~\ref{epaxos:exec:min}--\ref{epaxos:exec:set-exec}).
\end{enumerate}

\begin{example}
  Consider four commands $c_1$, $c_2$, $c_3$ and $c_4$, where $c_1$, $c_2$ and $c_3$ pairwise conflict, while $c_4$ conflicts only with $c_1$.
  Assume the commands are committed with dependencies that yield the graph in Figure~\ref{fig:exec_deps}, where $c \to c'$ means that $c$ is a dependency of $c'$.
  When the protocol in Figure~\ref{epaxos:exec} is executed on this graph, it will consider its three strongly connected components: $\{c_1\}$, $\{c_2, c_3\}$, and $\{c_4\}$.
  These can be topologically ordered in two ways: $\{c_1\} \rightarrow \{c_2, c_3\} \rightarrow \{c_4\}$ and $\{c_1\} \rightarrow \{c_4\} \rightarrow \{c_2, c_3\}$.
  The protocol will execute commands following one of these orders, breaking ties within strongly connected components using command identifiers.
  Thus, one process can use the order $c_1, c_2, c_3, c_4$, while the other can use $c_1, c_4, c_2, c_3$.
  This does not violate {\sf Consistency} because $c_4$ does not conflict with $c_2$ or $c_3$, and thus can be executed in either order with respect to them.
\end{example}

\begin{figure}[t]
  \center
  \begin{tikzpicture}[
  scale=1,
  every node/.style={transform shape},
  auto,
  shifttl/.style={shift={(-\shiftpoints,\shiftpoints)}},
  shifttr/.style={shift={(\shiftpoints,\shiftpoints)}},
  shiftbl/.style={shift={(-\shiftpoints,-\shiftpoints)}},
  shiftbr/.style={shift={(\shiftpoints,-\shiftpoints)}},
  node distance = 32pt
]
  
  \begin{scope}[transparency group]
    \begin{scope}[blend mode=multiply]
      \large

      \node [minimum size=15pt,Black,draw,fill=White,rounded corners=5] (c1) {$c_1$};
      \node [minimum size=15pt,Black,draw,fill=white,rounded corners=5, right of = c1,xshift=4pt] (c4) {$c_4$};
      \node [minimum size=15pt,Black,draw,fill=white,rounded corners=5, right of = c4,xshift=15pt] (c2) {$c_2$};
      \node [minimum size=15pt,Black,draw,fill=white,rounded corners=5, right of = c2,xshift=4pt] (c3) {$c_3$};

      \path[->,>=stealth, line width=1pt] (c1) edge[bend right] (c2);
      \path[->,>=stealth, line width=1pt] (c1) edge[bend left] (c3);
      \path[->,>=stealth, line width=1pt] (c1) edge (c4);
      \path[->,>=stealth, line width=1pt] (c2) edge[bend right] (c3);
      \path[->,>=stealth, line width=1pt] (c3) edge (c2);

    \end{scope}
  \end{scope}

\end{tikzpicture}
  \caption{Example of a dependency graph.}
  \label{fig:exec_deps}
\end{figure}
It is easy to show that, if Invariants~\ref{inv:agr}-\ref{inv:cnst} hold, then the protocol in Figure~\ref{epaxos:exec} ensures the {\sf Consistency} property of SMR (\refsection{smr});
we defer a formal proof to \tr{\ref{sec:proof}}{\nproof}.
Informally, Invariant~\ref{inv:agr} ensures that processes agree on the same dependency graph, while Invariant~\ref{inv:cnst} that the graph relates any two conflicting commands by at least one edge.
Thus, processes will use the same constraints to determine the order of execution for conflicting commands.

The original \epaxos protocol includes an optimization that speeds up the execution of commands that do not return a value: a process that submits such a command returns to the client as soon as the command is committed, without waiting for its dependencies to be committed too.
If applied straightforwardly, this would violate linearizability, so the protocol introduces additional mitigations that complicate it further.
As observed by follow-up work on \epaxos~\cite{epaxos-revisited}, commands without any return values at all are rare in real systems.
Hence, we do not consider this optimization in our protocol.

\subsection{Commit Protocol}
\label{sec:commit}

In \epp a client submits its command to one of the processes, called the {\it initial coordinator}, which will drive the agreement on the command's dependencies.
If the initial coordinator is suspected of failure, another process executes a {\it recovery} protocol to take over as a new coordinator.
For each command, a process follows only one coordinator at a time.
To ensure this, as in Paxos~\cite{paxos}, the lifecycle of the command is divided into a series of {\em ballots}, each managed by a designated coordinator.
A ballot is represented by a natural number, with ballot $0$ reserved for the initial coordinator.
Each process stores the current ballot number for a command $\id$ in $\bal[\id]$.
We now describe the protocol for committing commands at ballot $0$ (Figure~\ref{epaxos:cmt}).
This protocol closely follows the original \epaxos: the main differences for \epp are in recovery (\S\ref{sec:recover}).

\begin{figure}[t!]
\!\!\begin{minipage}{1.2\linewidth}
    \begin{minipage}{0.525\linewidth}
      \scalebox{0.85}{
      \begin{algorithm}[H]
        \renewcommand{\;}{\\}
        \SetKwProg{Function}{}{:}{}
        \DontPrintSemicolon

        \Function{$\submit(c)$\label{epaxos:cmt:submit}}{
          \Let $\id = \newid()$\;
          \label{epaxos:cmt:new-id}
          \send $\MPreAccept(\id, c, \{\id' \mid \inconflict{\cmd[\id']}{c}\})$\qquad \qquad \ToAll\;
          \label{epaxos:cmt:send-preacc}
        }

        \smallskip

        \SubAlgo{\onreceive $\MPreAccept(\id, c, D)$ \from $q$\label{epaxos:cmt:receive-preaccept}}{
          \precond $\bal[\id] = 0 \land \phase[\id] = \startP$\;
          \label{epaxos:cmt:preaccept-precond}
          $\assign{\cmd[\id]}{c}$\;
          \label{epaxos:cmt:set-c}
          $\assign{\initCmd[\id]}{c}$\;
          \label{epaxos:cmd:set-initcmd}
          $\assign{\initDep[\id]}{D}$\;
          \label{epaxos:cmt:set-d0}
          $\assign{\dep[\id]}{D \cup \{\id' \mid \inconflict{\cmd[\id']}{\cmd[\id]}\}}$\;
          \label{epaxos:cmt:preacc-deps}
          $\assign{\phase[\id]}{\preacceptP}$\;
          \label{epaxos:cmt:preacc-phase}
          \send $\MPreAcceptOK(\id, \dep[\id])$ \To $q$\;
        }

        \smallskip

        \SubAlgo{\onreceive $\MPreAcceptOK(\id, D_q)$ \qquad \qquad \qquad {\bf from all $q\in Q$}}{
          \label{epaxos:cmt:rec-preaccok}
          \precond $\bal[\id] \,{=}\, 0 \land \phase[\id] \,{=}\, \preacceptP \,{\land}$ \makebox[19pt]{}$|Q| \geq n-f$\;
          \label{epaxos:cmt:preaccok-precond}
          \Let $D = \bigcup_{q \in Q} D_q$\;
          \label{epaxos:cmt:let-d}
          \uIf{$|Q| \geq n\,{-}\,e \land \forall q \,{\in}\, Q.\,$\hlc{$D_q = \initDep[\id]$}}{
            \label{epaxos:cmt:fast-cmt}
            \send $\MCommit(0, \id, \cmd[\id], D)$ \ToAll\;
            \label{epaxos:cmt:send-cmt}
          }
          \Else{
            \label{epaxos:cmt:slow-cmt}
            \send $\MAccept(0, \id, \cmd[\id], D)$ \ToAll\;
            \label{epaxos:cmt:sent-acc}
          }
        }
      \end{algorithm}
      }
    \end{minipage}
    \hspace{-50pt}
      \begin{minipage}{0.55\linewidth}
        \scalebox{0.85}{
        \begin{algorithm}[H]
        \renewcommand{\;}{\\}
        \SetKwProg{Function}{}{:}{}
        \DontPrintSemicolon

        \SubAlgo{\onreceive $\MAccept(b, \id, c, D)$ \from $q$\label{epaxos:cmt:receive-accept}}{
          \precond $\bal[\id] \leq b \land {}$
          \makebox[18pt]{}$(\bal[\id] \,{=}\, b \,{\Longrightarrow}\, \phase[\id] \,{\not=}\, \commitP)$\;
          \label{epaxos:cmt:accept-precond}
          $\assign{\bal[\id]}{b}$\;
          \label{epaxos:cmt:accept-assign-b}
          $\assign{\cbal[\id]}{b}$\;
          \label{epaxos:cmt:accept-assign-cb}
          $\assign{\cmd[\id]}{c}$\;
          \label{epaxos:cmt:accept-assign-cmd}
          $\assign{\dep[\id]}{D}$\;
          $\assign{\phase[\id]}{\acceptP}$\; \label{epaxos:cmt:set-acceptp}
          \send $\MAcceptOK(b, \id)$ \To $q$\;
          \label{epaxos:cmt:sent-accok}
        }

        \smallskip

        \SubAlgo{\onreceive $\MAcceptOK(b, \id)$ {\bf from $Q$} \label{epaxos:cmt:receive-acceptok}}{
          \precond $\bal[\id] = b \land \phase[\id] = \acceptP \land {}$ \makebox[19pt]{}$|Q| \geq n-f$\;
          \label{epaxos:cmt:precond-cmt-slow}
          \send $\MCommit(b, \id, \cmd[\id], \dep[\id])$ \ToAll\;
          \label{epaxos:cmt:sent-cmt-slow}
        }

        \smallskip

        \SubAlgo{\onreceive $\MCommit(b, \id, c, D)$ \from $q$}{
          \label{epaxos:cmt:recv-cmt}
          \precond $\bal[\id] = b$\;
          $\assign{\cbal[\id]}{b}$\;
          \label{epaxos:cmt:cbal}
          $\assign{\cmd[\id]}{c}$\;
          \label{epaxos:cmt:assign-cmd}
          $\assign{\dep[\id]}{D}$\;
          $\assign{\phase[\id]}{\commitP}$\;
          \label{epaxos:cmt:set-phase}
        }

        \end{algorithm}
        }
        \smallskip
        \smallskip
        \smallskip
        \smallskip
        \smallskip
        \smallskip
        \smallskip
        \smallskip
      \end{minipage}

  \end{minipage}
  
  \caption{\epp: commit protocol.  Self-addressed messages are delivered
    immediately.  Initially, for all command identifiers $\id$:
    $\phase[\id]=\startP$, $\bal[\id]=\cbal[\id]=0$,
    $\initCmd[\id]=\cmd[\id]=\bot$, $\initDep[\id]=\dep[\id]=\emptyset$. The
    value $\bot$ does not conflict with any command.  }
  \label{epaxos:cmt}
\end{figure}

\subparagraph{Computing dependencies.}
When the initial coordinator receives a command $c$ from a client (line~\ref{epaxos:cmt:submit}), it first assigns to $c$ a unique identifier $\id$, from which the coordinator's identity can be extracted as $\coord(\id)$.
The coordinator then computes the {\em initial dependencies} of $c$ as the set of all conflicting commands it is aware of, and broadcasts them together with $c$ in a $\MPreAccept$ message.
Upon receiving a $\MPreAccept$ (line~\ref{epaxos:cmt:receive-preaccept}), a process records the command payload in $\cmd[\id]$ as well as $\initCmd[\id]$, and the dependencies computed by the initial coordinator in $\initDep[\id]$%
\footnote{As we noted earlier, in exceptional cases $\cmd[\id]$ may end up assigned to $\noop$.
  Keeping the originally submitted payload in $\initCmd[\id]$ is needed for the recovery protocol (\refsection{recover}).}.
The process then computes its own proposal for the dependency set by completing the initial dependencies with the set of conflicting commands it is aware of (line~\ref{epaxos:cmt:preacc-deps}).
Finally, it advances $\id$ to a $\preacceptP$ phase and replies to the coordinator with a $\MPreAcceptOK$ message, carrying its proposal.
After the coordinator receives $\MPreAcceptOK$ from a {\em quorum} $Q$ of at least $n-f$ processes (line~\ref{epaxos:cmt:rec-preaccok}), it decides whether to take the {\em fast path} or the {\em slow path}. The former commits the command immediately, while the latter requires an additional round of communication.

\subparagraph{Fast path.}
The coordinator takes the fast path (line~\ref{epaxos:cmt:fast-cmt}) if: {\em (i)} quorum $Q$ is {\em fast}, i.e., contains at least $n-e$ processes (recall that $e \le f$); and {\em (ii)} all the dependencies received from the processes in $Q$ are identical to the initial ones.
In this case the coordinator broadcasts a $\MCommit$ message, which informs all the processes about the agreed dependencies.
Upon receiving this message (line~\ref{epaxos:cmt:recv-cmt}), a process advances the command to the $\commitP$ phase, making it eligible for execution (\S\ref{sec:execution}).

For example, consider Figure~\ref{fig:runex}, where $n=5$, $f = 2$ and $e = 1$.
Process $p_1$ coordinates a command $c$ with identifier $\id$, and sends a $\MPreAccept$ message with $\emptyset$ as its initial dependencies.
Since this command is the first one to be received at $p_2$, $p_3$ and $p_4$, they reply with $\emptyset$.
Together with $p_1$, these processes form a fast quorum of size $n-e=4$.
Thus, $p_1$ immediately sends a $\MCommit$ message for $\id$ (depicted by dashed lines).

The fast path mechanism ensures that the protocol is $e$-fast, as defined in \refsection{smr}.
To see why, consider some $E$-faulty synchronous run with $|E| \leq e$.
Assume that a conflict-free command $c$ is submitted by a process $p$ at time $t$, and let $D_0$ be the initial dependencies that $p$ sends in a $\MPreAccept$ message for $c$.
Since $c$ is conflict-free, each conflicting command $c'$ either belongs to $D_0$, or cannot be submitted until $c$ is executed at some process.
Then every $\MPreAcceptOK$ message for $c$ sent back to $p$ will carry $D_0$.
Since at most $e$ processes fail, there will be $n-e$ such messages.
Then, since the run is synchronous, $p$ will execute $c$ by $t+2\Delta$.

\subparagraph{Slow path.}
If the coordinator does not manage to assemble a fast quorum agreeing on the dependencies, it takes a slow path (line~\ref{epaxos:cmt:slow-cmt}), which requires an additional round of communication.
The coordinator computes the final dependencies of the command as the union of the proposals it received in $\MPreAcceptOK$s and broadcasts them together with the initial dependencies in an $\MAccept$ message, analogous to the $\mathtt{2A}$ message of Paxos.
A recipient records the dependencies, advances the command to an $\acceptP$ phase, and replies to the coordinator with $\MAcceptOK$, analogous to the $\mathtt{2B}$ message of Paxos (line~\ref{epaxos:cmt:receive-accept}).
Once the coordinator receives a quorum of $\MAcceptOK$ messages (line~\ref{epaxos:cmt:receive-acceptok}), it broadcasts a $\MCommit$ message, informing the processes that the dependencies have been agreed.
The slow path code is reused at ballots $>0$, and thus, the above messages carry the ballot of the coordinator.
A process records the last ballot where it accepted a slow path proposal for each command in an $\cbal$ array (line~\ref{epaxos:cmt:accept-assign-cb}).
This corresponds to the second ballot variable of Paxos that was missing from the original \epaxos protocol (\S\ref{sec:intro}).

For example, in Figure~\ref{fig:runex} process $p_5$ coordinates a command $c'$ with identifier $\id'$, and sends a $\MPreAccept$ message with $\emptyset$ as its initial dependencies.
Since processes $p_3$ and $p_4$ receive this message after pre-accepting a conflicting command $c$ with identifier $\id$, they reply with $\{\id\}$.
Due to this disagreement, $p_5$ takes the slow path, and proposes $\{\id\}$ as
the dependency set for $\id'$.
This value is committed after another round of communication.
In the end, $c'$ is executed after $c$ (\S\ref{sec:execution}).

\begin{figure}[t!]
  \begin{tikzpicture}[decoration={
      markings,
      mark=at position 0.66 with {\arrow{>}},
      mark=at position 0.33 with {\arrow{>}}
    },
    scale=1.04,
    yscale=1.1, transform shape
  ]

  \filldraw
  (-3.6,-0.6) circle (0.7pt) node[align=center,above] {\footnotesize $p_1$}
  (-3.6,-1.2) circle (0.7pt) node[align=center,above] {\footnotesize $p_2$}
  (-3.6,-1.8) circle (0.7pt) node[align=center,above] {\footnotesize $p_3$}
  (-3.6,-2.4) circle (0.7pt) node[align=center,above] {\footnotesize $p_4$}
  (-3.6,-3)   circle (0.7pt) node[align=center,above] {\footnotesize $p_5$};

  \draw[Gray,line width=0.4pt] (-3.6,-0.6) -- (4.5,-0.6) node[scale=.75][xshift=.6em]{$\mathbf{//}$};
  \draw[Gray,->,line width=0.4pt,>=stealth] (4.8,-0.6) -- (9.5,-0.6);
  \draw[Gray,->,line width=0.4pt,>=stealth] (-3.6,-1.2) -- (9.5,-1.2);
  \draw[Gray,->,line width=0.4pt,>=stealth] (-3.6,-1.8) -- (9.5,-1.8);
  \draw[Gray,line width=0.4pt] (-3.6,-2.4) -- (6,-2.4) node[xshift=.5em, scale=1]{$\mathbf{\times}$};
  \draw[Gray,line width=0.4pt] (-3.6,-3) -- (6,-3) node[xshift=.5em, scale=1]{$\times$};

  \draw[Black,->,line width=0.1pt,>=stealth] (-3.4,-0.1) to (-3.2,-0.33);
  \draw[Black,->,line width=0.1pt,>=stealth] (-2.83,-3.4) to (-2.63,-3.17);
  
  \draw[Black,->,line width=0.5pt,>=stealth] (-2.5,-0.85) to (-1.6,-1.2);
  \draw[Black,->,line width=0.5pt,>=stealth] (-2.5,-0.85) to (-1.6,-1.8);
  \draw[Black,->,line width=0.5pt,>=stealth] (-2.5,-0.85) to (-1.6,-2.4);

  \draw[Black,->,line width=0.5pt,>=stealth] (-1.6,-1.2) to (-0.7,-0.85);
  \draw[Black,->,line width=0.5pt,>=stealth] (-1.6,-1.8) to (-0.7,-0.85);
  \draw[Black,->,line width=0.5pt,>=stealth] (-1.6,-2.4) to (-0.7,-0.85);

  \draw[->,line width=0.4pt,densely dashed,Gray,>=latex] (-0.7,-0.85) to (-0.6, -1.2);
  \draw[->,line width=0.4pt,densely dashed,Gray,>=latex] (-0.7,-0.85) to (-0.1, -1.8);
  \draw[->,line width=0.4pt,densely dashed,Gray,>=latex] (-0.7,-0.85) to (0.45, -2.4);
  \draw[->,line width=0.4pt,densely dashed,Gray,>=latex] (-0.7,-0.85) to (1, -3);

  \draw[Black,->,line width=0.5pt,>=stealth] (-1.8,-2.75) to (-0.8,-1.8);
  \draw[Black,->,line width=0.5pt,>=stealth] (-1.8,-2.75) to (-0.8,-2.4);

  \draw[Black,->,line width=0.5pt,>=stealth] (-0.8,-1.8) to (0.1, -2.75);
  \draw[Black,->,line width=0.5pt,>=stealth] (-0.8,-2.4) to (0.1, -2.75);

  \draw[Black,->,line width=0.5pt,>=stealth] (0.1, -2.75) to (1,-1.8);
  \draw[Black,->,line width=0.5pt,>=stealth] (0.1, -2.75) to (1,-2.4);

  \draw[Black,->,line width=0.5pt,>=stealth] (1,-1.8) to (1.9, -2.75);
  \draw[Black,->,line width=0.5pt,>=stealth] (1,-2.4) to (1.9, -2.75);

  \draw[Black,->,line width=0.1pt,>=stealth] (1.6,-0.1) to (1.8,-0.33); %
  \draw[Black,->,line width=0.1pt,>=stealth] (3.7,-0.1) to (4.2,-0.9); %

  \draw[Black,->,line width=0.5pt,>=stealth] (2.7,-0.85) to (3.3,-1.2);
  \draw[Black,->,line width=0.5pt,>=stealth] (3.3,-1.2) to (3.8,-0.6);

  \draw[Black,->,line width=0.5pt,>=stealth] (4.6,-1.2) to (5.35,-0.6);
  \draw[Black,->,line width=0.5pt,>=stealth] (4.6,-1.2) to (5.35,-1.8);

  \draw[Black,->,line width=0.5pt,>=stealth] (5.35,-0.6) to (6,-0.95);
  \draw[Black,->,line width=0.5pt,>=stealth] (5.35,-1.8) to (6,-1.45);

  \draw[Black,->,line width=0.5pt,>=stealth] (6,-0.95) to (6.65,-0.6);
  \draw[Black,->,line width=0.5pt,>=stealth] (6,-1.45) to (6.65,-1.8);

  \draw[Black,->,line width=0.5pt,>=stealth] (6.65,-0.6) to (7.3,-0.95);
  \draw[Black,->,line width=0.5pt,>=stealth] (6.65,-1.8) to (7.3,-1.45);

  \draw[Black,->,line width=0.5pt,>=stealth] (8.2,-0.95) to (8.7, -0.6);
  \draw[Black,->,line width=0.5pt,>=stealth] (8.2,-1.45) to (8.7,-1.8);

  \draw[->,line width=0.4pt,densely dashed,Gray,>=latex] (1.9, -2.75) to (6, -1.8);
  \draw[->,line width=0.4pt,densely dashed,Gray,>=latex] (1.9, -2.75) to (3, -2.4);

  \filldraw
  (-1.4,-1.04) node[inner sep=1pt,minimum size=1pt,Black,draw,fill=Green!20,right,line width=0.1pt,rounded corners=3] {\scalebox{0.8}{$\emptyset$}}
  (-1.4,-1.38) node[inner sep=1pt,minimum size=1pt,Black,draw,fill=Green!20,right,line width=0.1pt,rounded corners=3] {\scalebox{0.8}{$\emptyset$}}
  (-1.4,-1.8) node[inner sep=1pt,minimum size=1pt,Black,draw,fill=Green!20,right,line width=0.1pt,rounded corners=3] {\scalebox{0.8}{$\emptyset$}}

  (-0.7,-2.1) node[inner sep=1.2pt,minimum size=1pt,Black,draw,fill=Red!20,right,line width=0.1pt,rounded corners=3] {\scalebox{0.7}{$\{\id\}$}}
  (-0.7,-2.53) node[inner sep=1.2pt,minimum size=1pt,Black,draw,fill=Red!20,right,line width=0.1pt,rounded corners=3] {\scalebox{0.7}{$\{\id\}$}}

  (3.4,-0.95) node[inner sep=1pt,minimum size=1pt,Black,draw,fill=Blue!20,right,line width=0.1pt,rounded corners=3] {\scalebox{0.7}{$\{\id\}$}}

  (8.4,-0.4) node[inner sep=1pt,minimum size=1pt,fill=White,right,line width=0.1pt,rounded corners=3] {\scalebox{0.8}{$\noop$}}
  (8.4,-2) node[inner sep=1pt,minimum size=1pt,fill=White,right,line width=0.1pt,rounded corners=3] {\scalebox{0.8}{$\noop$}}
  ;

  \filldraw
  (-3.3,-0.6) node[minimum height=14pt,Black,fill=Green!20,draw,right] {\scalebox{0.8}{$\MPreAccept$}}
  (-1.3,-0.6) node[minimum height=14pt,Black,fill=Green!20,draw,right] {\scalebox{0.8}{$\MCommit$}}
  (-2.6,-3) node[minimum height=14pt,Black,fill=Red!20,draw,right] {\scalebox{0.8}{$\MPreAccept$}}
  (-0.5,-3) node[minimum height=14pt,Black,fill=Red!20,draw,right] {\scalebox{0.8}{$\MAccept$}}
  (1.3,-3) node[minimum height=14pt,Black,fill=Red!20,draw,right] {\scalebox{0.8}{$\MCommit$}}
  (1.7,-0.6) node[minimum height=14pt,Black,fill=Blue!20,draw,right] {\scalebox{0.8}{$\MPreAccept$}}
  (4,-1.2) node[minimum height=14pt,Black,fill=Blue!20,draw,right] {\scalebox{0.8}{$\MRecover$}}
  (5.4,-1.2) node[minimum height=14pt,Black,fill=Blue!20,draw,right] {\scalebox{0.8}{$\MValidate$}}
  (7.1,-1.2) node[minimum height=14pt,Black,fill=Blue!20,draw,right] {\scalebox{0.8}{$\MAccept$}};

  \filldraw
  (-3.1,0) node{\scalebox{0.7}{$\submit(\cmdc)$}}
  (-2.45,-3.55) node{\scalebox{0.7}{$\submit(\cmdcp)$}}
  (0.64,-0.6) node[above] {\scalebox{0.75}{$\exec(\cmdc)$}}
  (3.4,-3) node[above] {\scalebox{0.75}{$\exec(\cmdc)$}}
  (4.9,-3) node[above] {\scalebox{0.75}{$\exec(\cmdcp)$}}
  (1.9,0) node{\scalebox{0.7}{$\submit(\cmdcpp)$}}
  (4.2,0) node{\scalebox{0.7}{$\recover(\id'')$}};
\end{tikzpicture}
  \caption{Example of an \epp run with pairwise conflicting commands $c$, $c'$,
    and $c''$, which have identifiers $\id$, $\id'$ and $\id''$, respectively.
    We assume $n = 5$, $f = 2$ and $e = 1$.
    During the run, $p_1$ is briefly disconnected ({\color{Gray}{$//$}}), while $p_4$ and $p_5$ fail ({\color{Gray}{$\times$}}).
  }
  \label{fig:runex}
\end{figure}

\subparagraph{Ordering guarantee.}
A key property of the above protocol is that it preserves Invariant~\ref{inv:cnst} for pairs of conflicting commands committed at ballot $0$: at least one of the two commands is in the dependency set of the other.
\begin{proof}[Proof of Invariant~\ref{inv:cnst} for ballot $0$]
  Take two distinct commands $\id$ and $\id'$, committed at ballot $0$ with payloads $c$ and $c'$ and dependency sets $D$ and $D'$, respectively.
  Assume that $\inconflict{c}{c'}$.
  According to line~\ref{epaxos:cmt:let-d}, there is a quorum $Q$ of $\MPreAcceptOK(\id, D_q)$ messages such that $|Q| \geq n-f$ and $D = \bigcup_{q \in Q} D_q$.
  Similarly, there is a quorum $Q'$ of $\MPreAcceptOK(\id', D'_q)$ messages such that $|Q'| \geq n-f$ and $D' = \bigcup_{q \in Q'} D'_q$.
  By line~\ref{epaxos:cmt:set-c}, before sending $\MPreAcceptOK$, processes in $Q$ (respectively, $Q'$) store $c$ (respectively, $c'$) in $\cmd$.
  Since $n \geq 2f+1$, we have $|Q \cap Q'| \geq n-2f \geq 1$, and thus $Q \cap Q' \neq \emptyset$.
  Take any process $q \in Q \cap Q'$ and assume without loss of generality that $q$ sends $\MPreAcceptOK(\id, D_q)$ before $\MPreAcceptOK(\id', D'_q)$.
  Then the computation at line~\ref{epaxos:cmt:preacc-deps} ensures that $\id \in D'_q$, and thus $\id \in D'$.
\end{proof}

\subsection{Recovery Protocol}
\label{sec:recover}

If a coordinator of a command is suspected of failure, another process is nominated to take over;
this nomination is done using standard techniques based on failure detectors~\cite{omega}, and we defer the details to \tr{\ref{app:start-recovery}}{\nliveness}.
The new coordinator executes a recovery protocol for a command with identifier $\id$ by calling $\recover(\id)$ in Figure~\ref{epaxos:recvr} (for now, we ignore the highlighted parts of the code, which are explained in \S\ref{sec:coord-opt}).
The coordinator picks a ballot it owns higher than any it has previously participated in and asks processes to join it with a $\MRecover$ message, analogous to the $\mathtt{1A}$ message of Paxos.
Processes join the ballot if they do not already participate in a higher one.
In this case they reply with a $\MRecoverOK$ message, which is analogous to the $\mathtt{1B}$ message of Paxos and carries all the local information about the command being recovered (line~\ref{epaxos:recvr:rec}).
The new coordinator waits until it receives $\MRecoverOK$ messages from a quorum $Q$ of at least $n-f$ processes (line~\ref{epaxos:recvr:recoverok}), after which it computes its proposal for the dependency set.
To maintain consensus on dependencies (Invariant~\ref{inv:agr}), the new coordinator must propose the same dependencies that were decided in earlier ballots, if any.
The latter could be done either on the slow or the fast path, and the protocol addresses the two cases separately.

\begin{figure}[p]
  \begin{minipage}{1.23\linewidth}
    \scalebox{0.85}{
      \begin{algorithm}[H]
        \renewcommand{\;}{\\}
        \SetKwProg{Function}{}{:}{}
        \DontPrintSemicolon

        \Function{$\recover(\id)$\label{epaxos:recvr:recover}}{
          \Let {\ $b$ = (a ballot owned by $p$ such that $b > \bal[\id]$)}\;
          \label{epaxos:recvr:new-bal}
          \send $\MRecover(b, \id)$ \ToAll{}
          \label{epaxos:recvr:send-rec}
        }

        \smallskip
        \smallskip

        \SubAlgo{\onreceive $\MRecover(b, \id)$ \from $q$\label{epaxos:recvr:rec}}{
          \precond $\bal[\id] < b$\;
          \label{epaxos:recvr:rec-precond}
          $\assign{\bal[\id]}{b}$\;
          \label{epaxos:recvr:set-bal}
          \send $\MRecoverOK(b, \id, \cbal[\id], \cmd[\id], \dep[\id], \initDep[\id], \phase[\id])$ \To $q$\;
          \label{epaxos:recvr:send-recok}
        }

        \smallskip
        \smallskip

        \SubAlgo{\onreceive $\MRecoverOK(b,\id,\vcbal_q,c_q,\vdep_q,\vinitDep_q,\vphase_q)$ {\bf from all $q \in Q$}\label{epaxos:recvr:recoverok}}{
          \precond $\bal[\id] = b \land |Q| \geq n-f$\;
          \label{epaxos:recvr:recoverok-precond}
          \Let $\bmax = {\tt max}\{\vcbal_q \mid q \in Q\}$\; \label{epaxos:recvr:bmax}
          \Let $U = \{q \in Q \mid \vcbal_q = \bmax\}$\;
          \label{epaxos:recvr:let-u}

          \lIf{$\exists q \in U.\, \vphase_q = \commitP$}{\label{epaxos:recvr:is-cmt}
            \send $\MCommit(b, \id, c_q, \vdep_q)$ \ToAll
            \label{epaxos:recvr:send-cmt}
          }
          \lElseIf{$\exists q \in U.\, \vphase_q = \acceptP$}{\label{epaxos:recvr:is-acc}
            \send $\MAccept(b, \id, c_q, \vdep_q)$ \ToAll
            \label{epaxos:recvr:send-acc}
          }
          \SetAlgoNoEnd
          \hl{\bf else if $\coord(\id) \in Q$ then send $\MAccept(b, \id, \noop, \emptyset)$ to all}\;
          \label{epaxos:recvr:coord-in}
          \label{epaxos:recvr:send-acc-coord}

          \uElseIf{\hlc{$\exists R \subseteq Q.\, |R| \geq |Q|-e \land \forall q \in R.\, (\vphase_q = \preacceptP \land \vdep_q = \vinitDep_q)$}}{
            \label{epaxos:recvr:fast-cond}
            \hl{{\bf let} $\Rmax$ be the largest set $R$ that satisfies the condition at line~\mbox{\ref{epaxos:recvr:fast-cond}}}\;
            \label{epaxos:recvr:let-rmax}
            \Let {\ $(c, D) = (c_q, \vdep_q)$ for any $q \in R$}\;
            \label{epaxos:recvr:let-c}
            \send $\MValidate(b, \id, c, D)$ {\bf to all processes in $Q$}\;
            \label{epaxos:recvr:send-validate}
            {\bf wait until received} $\MValidateOK(b, \id, \conflicts_q)$ {\bf from all $q \in Q$}\;
            \label{epaxos:recvr:rec-validate-ok}
            \Let $\conflicts = \bigcup_{q \in Q} \conflicts_q$\;
            \label{epaxos:recvr:rec-let-conflicts}
            \uIf{$\conflicts = \emptyset$}{\label{epaxos:recvr:conf-empty}
              \send $\MAccept(b, \id, c, D)$ \ToAll
              \label{epaxos:recvr:send-acc-good}
            }
            \uElseIf{$(\exists (\id',\commitP) \in \conflicts)$ \label{epaxos:recvr:conf-cmt}
              $\text{\hl{$\lor\ (|\Rmax| = |Q| - e
                \land \exists (\id', \_) \in \conflicts.\,
                \coord(\id') \notin Q)$}}$}{
              \label{epaxos:recvr:break-cond-wait} %
              \send $\MAccept(b, \id, \noop, \emptyset)$ \ToAll
              \label{epaxos:recvr:conf-cmt-acc}
              \label{epaxos:recvr:nop-3} %
            }
            \Else{
              \send \hlc{$\MWaiting(\id, \text{\hl{$|\Rmax|$}})$} \ToAll\;
              \label{epaxos:recvr:send-wait}
              \SetAlgoNoLine
              \SetAlgoNoEnd
              \SubAlgo{\bf wait until}{
                \label{epaxos:recvr:wait}
                \SetAlgoVlined
                \SetAlgoNoEnd
                \SubAlgo{{\bf case $\exists (\id', \_) \in \conflicts.\, \phase[\id'] = \commitP \land (\cmd[\id'] \neq \noop \land \id \notin \dep[\id'])$ do}}{
                  \label{epaxos:recvr:break-cond-cmt}
                  \send $\MAccept(b, \id, \noop, \emptyset)$ \ToAll\;
                  \label{epaxos:recvr:nop-1}
                }
                \SubAlgo{{\bf case $\forall (\id', \_) \in \conflicts.\, \phase[\id'] = \commitP \land (\cmd[\id'] = \noop \lor \id \in \dep[\id'])$ do}}{
                  \label{epaxos:recvr:all-committed}
                  \send $\MAccept(b, \id, c, D)$ \ToAll\;
                  \label{epaxos:recvr:nop-2}
                }
                \SubAlgo{{\bf case $\exists (\id', \_) \in \conflicts.\, (p$
                    received \hlc{$\MWaiting(\id', \text{\hl{$k'$}}))
                      \ \text{\hl{$\land\ k' > n-f-e$}}$} do}}{
                  \label{epaxos:recvr:break-cond-wait-r}
                  \send $\MAccept(b, \id, \noop, \emptyset)$ \ToAll\;
                  \label{epaxos:recvr:break-cond-wait-r-end}
                }
                  \SubAlgo{\hl{\bf case $p$ received $\MRecoverOK(b,\id,\_,\vcmd,\vdep,\_,\vphase)$ from $q \notin Q$
                      with $\vphase = \commitP \lor \vphase = \acceptP \lor
                      q = \coord(\id)$ do}}{
                    \label{epaxos:recvr:break-cond-coord}
                    \hl{\bf if $\vphase = \commitP$ then send $\MCommit(b, \id, \vcmd, \vdep)$ to all}\; \label{epaxos:recvr:coord-commit}
                    \hl{\bf else if $\vphase = \acceptP$ then send $\MAccept(b, \id, \vcmd, \vdep)$ to all}\; \label{epaxos:recvr:coord-acc}
                    \hl{\bf else send $\MAccept(b, \id, \noop, \emptyset)$ to all}
                    \label{epaxos:recvr:break-cond-coord-end}
                  }
              }
            }
          }
          \lElse{\send $\MAccept(b, \id, \noop, \emptyset)$ \ToAll}
          \label{epaxos:recvr:last-acc}
        }

        \smallskip
        \smallskip

        \SubAlgo{\onreceive $\MValidate(b, \id, c, D)$ \from $q$}{
          \precond $\bal[\id] = b$\;
          $\assign{\cmd[\id]}{c}$\;
          \label{epaxos:recvr:set-c}
          $\assign{\initCmd[\id]}{c}$\;
          \label{epaxos:recvr:set-initcmd}
          $\assign{\initDep[\id]}{D}$\;
          \label{epaxos:recvr:set-d0}

          \Let $\conflicts = \{(\id', \phase[\id']) \mid \id' \neq \id \land \id' \notin D \land{}$
          \makebox[51pt]{}$(\phase[\id'] = \commitP \Longrightarrow {\cmd[\id'] \neq \noop} \land {\inconflict{\cmd[\id']}{c}} \land {\id \notin \dep[\id']}) \land{}$
          \makebox[51pt]{}$(\phase[\id'] \neq \commitP \Longrightarrow {\initCmd[\id'] \neq \bot} \land {\inconflict{\initCmd[\id']}{c}} \land {\id \notin \initDep[\id']})\}$\;
          \label{epaxos:recvr:let-conflicts}
          \send $\MValidateOK(b, \id, \conflicts)$ \ToAll\:
          \label{epaxos:recvr:send-valok}
        }
        \smallskip
      \end{algorithm}
    }
  \end{minipage}
  \caption{\epp: recovery protocol at a process $p$.
    Self-addressed messages are delivered immediately.
    The highlighted parts of the code enable the optimization of \S\ref{sec:coord-opt}.
  }
  \label{epaxos:recvr}
\end{figure}

\subparagraph{Recovering slow path decisions.}

As in Paxos, the new coordinator first considers the subset $U \subseteq Q$ of processes with the highest reported $\cbal$ (line~\ref{epaxos:recvr:let-u}, assigned at line~\ref{epaxos:cmt:accept-assign-cb}): the states of these processes supersede those from lower ballots.
If a process in $U$ has already committed the command $\id$ (line~\ref{epaxos:recvr:is-cmt}), then the coordinator broadcasts a $\MCommit$ message with the committed dependency set.
Similarly, if a process in $U$ has accepted the command $\id$ (line~\ref{epaxos:recvr:is-acc}), then the coordinator resumes the commit protocol interrupted by its predecessor by broadcasting an $\MAccept$ message with the accepted dependency set and the new ballot.
These rules are enough to recover dependencies committed via the slow path: in this case there must exist a quorum of $\ge n-f$ processes each having $\phase[\id] \in \{\acceptP, \commitP\}$; since $n \ge 2f+1$, this quorum intersects with the recovery quorum $Q$, so one of the conditions at lines~\ref{epaxos:recvr:is-cmt} or~\ref{epaxos:recvr:is-acc} must hold.

\subparagraph{Recovering fast path decisions.}
Even if none of the processes in $U$ has committed or accepted $\id$, it is still possible that the initial coordinator committed the command via the fast path.
The recovery of such decisions is much more subtle.
A naive approach would be to do this similarly to Fast Paxos~\cite{fast-paxos}.
If $\id$ took the fast path then, since the size of the fast quorum is $n-e$ (line~\ref{epaxos:cmt:fast-cmt}), at least $|Q|-e$ processes in the recovery quorum $Q$ must have pre-accepted $\id$'s initial dependencies.
Thus, if the latter condition is satisfied (line~\ref{epaxos:recvr:fast-cond}), the new coordinator could recover the command $\id$ with its initial dependencies by broadcasting them in an $\MAccept$ message in the new ballot (line~\ref{epaxos:recvr:send-acc-good}).
Otherwise, the coordinator would {\em abort} the recovery of $\id$ by proposing $\noop$ as its payload in the new ballot (line~\ref{epaxos:recvr:last-acc}).
This approach preserves Invariant~\ref{inv:agr}, but unfortunately, it breaks Invariant~\ref{inv:cnst}, as we now illustrate.

\begin{example}
  \label{ex:break-deps-a}
  Consider again Figure~\ref{fig:runex}.
  After commands $\id$ and $\id'$ are committed with payloads $c$ and $c'$, process $p_1$ submits another command $\id''$, with payload $c''$ that conflicts with $c$ and $c'$.
  The initial dependency set of $\id''$ is $\{\id\}$.
  Since $p_2$ did not participate in committing $\id'$, it pre-accepts $\id''$ with its initial dependencies.
  At this point, suppose the network connection at $p_1$ is briefly disrupted.
  Process $p_2$ suspects that $p_1$ failed and starts recovering $\id''$.
  If it gathers a quorum $Q_{\id''} = \{p_1, p_2, p_3\}$, it will observe itself and $p_1$ voting for the same dependency set $\{\id\}$.
  Together with $p_4$ and $p_5$ these processes {\em could} have formed a fast quorum $\{p_1, p_2, p_4, p_5\}$ to commit $\id''$ with $\{\id\}$.
  Thus, $p_2$ may think that it must recover $\id''$ with the dependency set $\{\id\}$ to preserve Invariant~\ref{inv:agr}.
  But then the pair of commands $\id'$ and $\id''$ breaks Invariant~\ref{inv:cnst}, and thus, {\sf Consistency} (\S\ref{sec:smr}).
\end{example}

\subparagraph{Validation phase overview.}
The above pitfall arises because the new coordinator can only {\em suspect} that a command took the fast path, since it observes only a part of the fast quorum: $|Q|-e$ votes out of $n-e$.
The coordinator can then resurrect commands (such as $\id''$ in Example~\ref{ex:break-deps-a}) that did not actually take the fast path and are missing from the dependencies of conflicting commands (such as $\id'$).
To avoid this, \epp augments the above naive recovery with a special {\em validation phase}, which performs a more precise analysis of whether the command being recovered could have taken the fast path.
This validation phase is the key novel contribution of our protocol: as we explain in \S\ref{sec:comparison}, it is simpler than the original \epaxos recovery, and has been rigorously proved correct \tra{\ref{sec:proof}}{\nproof}.

During the validation phase, the new coordinator sends to its recovery quorum $Q$ a $\MValidate$ message with the payload $c$ and the dependencies $D$ that it would like to propose in the new ballot for the command being recovered (line~\ref{epaxos:recvr:send-validate}).
The recipients reply with $\MValidateOK$ messages carrying additional information (line~\ref{epaxos:recvr:send-valok}).
To see what kind of information is helpful to the new coordinator, consider Example~\ref{ex:break-deps-a} again.
Process $p_3 \in Q_{\id''}$ committed $\id'$ with dependencies that do not contain $\id''$, and thus it knows that $\id''$ could not take the fast path with a dependency set that does not contain $\id'$: this would contradict Invariant~\ref{inv:cnst} for commands committed at ballot $0$, which always holds, as we proved in \S\ref{sec:commit}.
In this case, we say that command $\id'$ {\em invalidates} the recovery of $\id''$.
The $\MValidateOK$ messages carry a set $I$ including the identifiers of such commands (the first implication at line~\ref{epaxos:recvr:let-conflicts}).
\begin{definition}
  \label{def:inv}
  A command $\id'$ \textbf{invalidates} the recovery of a command $\id$ with a payload $c$ and a dependency set $D$, if a process has committed $\id'$ with a payload $c' \neq \noop$ and a dependency set $D'$ such that $\inconflict{c'}{c}$, $\id' \notin D$ and $\id \notin D'$.
\end{definition}

As illustrated by the above example, recovering $\id$ with $c$ and $D$ in the presence of an invalidating command would violate Invariant~\ref{inv:cnst}.
Hence, in this case the coordinator aborts the recovery of $\id$ by replacing its payload with a $\noop$ command.
This does not violate Invariant~\ref{inv:agr} because in this case we can prove that it could not take the fast path.
For instance, if the invalidating command is committed at ballot $0$, then this follows from the proof of Invariant~\ref{inv:cnst} for ballot $0$ given in \S\ref{sec:commit}.
The set carried by $\MValidateOK$ messages also includes the identifiers of {\em potentially invalidating} commands -- those that are not yet committed but may become invalidating once they are (the second implication at line~\ref{epaxos:recvr:let-conflicts}).
\begin{definition}
  \label{def:pinv}
    A command $\id'$ \textbf{potentially invalidates} the recovery of a command $\id$ with a payload $c$ and a dependency set $D$, if a process has not committed $\id'$, but stores it with an initial payload $c'$ and an initial dependency set $D_0'$ such that $\inconflict{c'}{c}$, $\id' \notin D$ and $\id \notin D_0'$.
\end{definition}

We can prove that any uncommitted command $\id'$ that may eventually invalidate $\id$ must satisfy this definition.
In particular, note that the initial dependencies $D_0'$ of $\id'$ are always included into its committed dependencies (line~\ref{epaxos:cmt:preacc-deps}).
Hence, if $\id \in D_0'$ (contradicting Definition~\ref{def:pinv}), then $\id'$ cannot be committed with dependencies $D'$ such that $\id \not\in D'$, and thus $\id'$ cannot invalidate $\id$.
If the coordinator finds potentially invalidating commands, it cannot yet determine whether it is safe to recover $\id$.
Thus, it waits for those commands to be committed before making a decision.

Note that the second disjunct of line~\ref{epaxos:recvr:let-conflicts} and Definition~\ref{def:pinv} consider a command $\id'$ to be potentially invalidating for $\id$ even if a process has $\cmd[\id'] = \noop$, but
$\initCmd[\id'] \not= \bot$ and $\inconflict{\initCmd[\id']}{c}$.
This is because we cannot yet be sure that $\id'$ will be committed with a $\noop$: if the coordinator trying to do this fails, another coordinator may query a different recovery quorum and decide to propose the original payload $\initCmd[\id']$.

\subparagraph{Validation details.}
In more detail, upon receiving $\MValidateOK$ messages from the recovery quorum $Q$ (line~\ref{epaxos:recvr:rec-validate-ok}), the new coordinator acts as follows:
\begin{enumerate}%
\item
  \label{fast-rec-case-1}
  If no invalidating commands are detected at the quorum, the coordinator recovers $\id$ with $c$ and $D$ (line~\ref{epaxos:recvr:conf-empty}).
  In this case, no invalidating command will appear in the future: since the payload $c$ is now stored at the quorum (line~\ref{epaxos:recvr:set-c}), any future conflicting command will depend on $\id$.
\item
  If some invalidating command is detected, then the coordinator aborts the recovery of $\id$ by replacing its payload with a $\noop$ command (line~\ref{epaxos:recvr:conf-cmt}).
  \label{fast-rec-case-2}
\item
  \label{fast-rec-case-3}
  If none of the above holds, the coordinator waits until the potentially invalidating commands are committed. %
  If some of these commands end up being invalidating, just like in item~\ref{fast-rec-case-2}, the coordinator aborts the recovery of $\id$ (line~\ref{epaxos:recvr:break-cond-cmt}).
  Otherwise, the coordinator recovers $\id$ with $c$ and $D$ (line~\ref{epaxos:recvr:all-committed}).
\item
  Item~\ref{fast-rec-case-3} may require the coordinator of $\id$ to wait for a potentially invalidating command $\id'$ that is itself undergoing recovery.
  The recovery of $\id'$ may then also need to wait for another command.
  One may get worried that this will lead to cycles of commands waiting for each other, making the protocol stuck.
  Fortunately, if the coordinator of $\id$ finds out that $\id'$ is itself blocked, the coordinator can stop waiting for $\id'$.
  To enable this, before any process starts the wait at line~\ref{epaxos:recvr:wait}, it broadcasts a $\mathtt{Waiting}$ message, carrying the identifier of the command it is recovering (line~\ref{epaxos:recvr:send-wait}). %
  While the coordinator of $\id$ is waiting, it listens for $\mathtt{Waiting}$ messages from other coordinators.
  Upon receiving $\mathtt{Waiting}(\id')$, the coordinator aborts the recovery of $\id$ (line~\ref{epaxos:recvr:break-cond-wait-r}).
  The following proposition implies that this is safe because, in this case, $\id$ could not take the fast path.
\end{enumerate}
\begin{lemma}
  \label{prop:wait}
  Consider commands $\id$ and $\id'$ with conflicting initial payloads $c$ and $c'$, and assume that  neither of the commands is included in the initial dependencies of the other.
  If a message $\mathtt{Waiting}(\id')$ has been sent, then $\id$ did not take the fast path.
\end{lemma}
\begin{proof}
  By contradiction, assume that $\id$ takes the fast path with a fast quorum $Q_f$, i.e., $|Q_f| \geq n-e$ and all processes in $Q_f$ pre-accepted $\id$ with its initial dependencies $D$ (line~\ref{epaxos:cmt:fast-cmt}).
  Since $\mathtt{Waiting}(\id')$ has been sent, there exists a recovery quorum $Q'$ and a set $R' \subseteq Q'$ such that $|R'| \geq |Q'|-e \geq n-f-e$ and all processes in $R'$ pre-accepted $\id'$ with its initial dependencies $D'$ (line~\ref{epaxos:recvr:fast-cond}).
  Since $n \geq 2e+f+1$, the intersection of $Q_f$ and $R'$ has a cardinality
\begin{equation}\label{eq:intersect}
  \ge (n-e)+(n-f-e)-n = n -2e -f \ge 1.
\end{equation}
  Thus, some process $p$ pre-accepted $\id$ with $D$ and $\id'$ with $D'$.
  Assume without loss of generality that the former happens before the latter.
  When $p$ pre-accepts $\id$, it assigns $\cmd[\id]$ to $c$ (line~\ref{epaxos:cmt:set-c}).
  It is easy to see that, at any time after this, $p$ has $\cmd[\id] = c$ or $\cmd[\id] = \noop$.
  In particular, this holds when $p$ pre-accepts $\id'$.
  Then, since $\inconflict{c}{c'}$ and $\inconflict{\noop}{c'}$, line~\ref{epaxos:cmt:preacc-deps} implies $\id \in D'$, which contradicts our assumption.
\end{proof}

Using this result, in \tr{\ref{sec:proof}}{C} we prove the correctness of our protocol.
\begin{theorem}
  \label{theo:upper-unopt}
  The above-presented {\em \epp} protocol implements an $f$-resilient $e$-fast SMR protocol, provided $n \ge \max\{2e+f+1, 2f+1\}$.
\end{theorem}

\section{\epp: Optimized Protocol for {\boldmath $n \geq \max\{2e+f-1, 2f+1\}$}}
\labsection{coord-opt}

We now explain how to optimize the protocol presented so far to require $n \geq \max\{2e+f-1, 2f+1\}$, which matches the lower bound in Theorem~\ref{theo:lower}.
For $n = 2f+1$ this allows setting $e = \lceil\frac{f+1}{2}\rceil$ -- the threshold the original \epaxos aimed to achieve~\cite{epaxos}.
In particular, for the values $n = 3$ and $n=5$, frequent in practice, $e$ can be a minority of processes.
In these cases, the fast path contacts quorums of the same size as in the usual, slower Paxos~\cite{paxos}.

The optimization is enabled by adding the highlighted code to the recovery protocol in Figure~\ref{epaxos:recvr}.
First, before attempting to recover a fast path decision for a command $\id$, the new coordinator checks whether the initial coordinator of $\id$ belongs to the recovery quorum (line~\ref{epaxos:recvr:coord-in}) -- if so, it aborts the recovery, proposing a $\noop$.
This is safe because in this case $\id$ cannot take the fast path.
Indeed, the initial coordinator of $\id$ could not take the fast path before sending its $\MRecoverOK$ message: otherwise, it would report the command as committed, and the new coordinator would take the branch at line~\ref{epaxos:recvr:is-cmt}.
The initial coordinator will also not take the fast path in the future: when sending $\MRecoverOK$, it switches to a ballot $>0$, which disables the fast path (line~\ref{epaxos:cmt:preaccok-precond}).

We also make a change to the case when the new coordinator waits for potentially invalidating commands to commit (line~\ref{epaxos:recvr:wait}).
If during this wait the new coordinator receives an additional $\MRecoverOK$ message for $\id$ from a process that has committed or accepted $\id$, or from the initial coordinator of $\id$ (line~\ref{epaxos:recvr:break-cond-coord}), it handles the message just like a previously received one, completing the recovery as in lines~\ref{epaxos:recvr:is-cmt}--\ref{epaxos:recvr:send-acc-coord}.

The next change is that the new coordinator of $\id$ aborts the recovery if exactly $|Q|-e$ processes in the recovery quorum $Q$ satisfy the condition at line~\ref{epaxos:recvr:fast-cond}, and the coordinator of a potentially invalidating non-$\noop$ command is not part of $Q$ (line~\ref{epaxos:recvr:conf-cmt}).
The following proposition shows that this is safe.
\begin{lemma}
  \label{prop:red-one}
  If the first disjunct at line~\ref{epaxos:recvr:conf-cmt} does not hold, but the second one does, then $\id$ did not take the fast path.
\end{lemma}
\begin{proof}
  By contradiction, assume that $\id$ takes the fast path with a fast quorum $Q_f$, i.e., $|Q_f| \geq n-e$ and all processes in $Q_f$ pre-accepted $\id$ with its initial dependencies $D$ (line~\ref{epaxos:cmt:fast-cmt}).
  By our assumption, $|\Rmax| = |Q|-e$ and for some $\id'$ and $\vphase'$ we have $(\id', \vphase') \in \conflicts$, $\coord(\id') \notin Q$, and $\vphase' \neq \commitP$.
  By line~\ref{epaxos:recvr:let-conflicts}, we also have $\id' \not\in D$.
  Let $D'$ be the initial dependency set of $\id'$.
  Notice that by lines~\ref{epaxos:cmt:set-d0} and \ref{epaxos:recvr:set-d0}, whenever a process has $\initCmd[\id'] \neq \bot$, it also has $\initDep[\id'] = D'$.
  Then by line~\ref{epaxos:recvr:let-conflicts}, the process whose $\MValidateOK$ message resulted in
  $(\id', \vphase')$ being included into $I$ at line~\ref{epaxos:recvr:rec-let-conflicts} must have had $\id \notin \initDep[\id'] = D'$.
  Moreover, by the same line, $\id$ and $\id'$ have conflicting initial payloads.
  Since $|\Rmax| = |Q|-e$, we know that $|Q|-e$ processes in $Q$ pre-accepted $\id$ with $D$, and the remaining $e$ did not.
  Since $|Q_f| \geq n-e$ processes pre-accepted $\id$ with $D$, all processes that are not in $Q$ must belong to $Q_f$.
  Then since $\coord(\id') \notin Q$, we must have $\coord(\id') \in Q_f$, so $\coord(\id')$ pre-accepted $\id$ with $D$ and $\id'$ with $D'$.
  Given that the initial payloads of $\id$ and $\id'$ conflict, we obtain a contradiction in the same way as in the proof of Lemma~\ref{prop:wait}.
\end{proof}

Finally, in the baseline protocol the new coordinator of $\id$ could abort the recovery if it received a $\mathtt{Waiting}(\id')$ message for a potentially invalidating command $\id'$ (line~\ref{epaxos:recvr:break-cond-wait-r}).
This cannot generally be done in the optimized protocol, because Lemma~\ref{prop:wait} no longer holds: its proof relies on the fact that any two sets of $n-e$ and $n-f-e$ processes intersect (Eq.~\ref{eq:intersect}), which is not ensured with $n \ge \max\{2e + f - 1, 2f+1\}$.
Hence, the optimized protocol imposes an additional constraint on when the new coordinator can abort the recovery.
Each $\MWaiting$ message sent at line~\ref{epaxos:recvr:send-wait} includes the size of the largest set $R$ satisfying the condition at line~\ref{epaxos:recvr:fast-cond}.
Then if the new coordinator of $\id$ receives $\mathtt{Waiting}(\id', k')$, it can only abort the recovery if $k' > n - f - e$ (note that we always have $k' \ge n - f - e$).
The following weaker version of Lemma~\ref{prop:wait} shows that this is safe.
\begin{lemma}
  \label{prop:wait-opt}
  Consider commands $\id$ and $\id'$ with conflicting initial payloads $c$ and $c'$, and assume that  neither of the commands is included in the initial dependencies of the other.
  If a message $\mathtt{Waiting}(\id', k')$ has been sent with $k' > n-f-e$, then $\id$ did not take the fast path.
\end{lemma}
\begin{proof}
  By contradiction, assume that $\id$ takes the fast path with a fast quorum $Q_f$, i.e., $|Q_f| \geq n-e$ and all processes in $Q_f$ pre-accepted $\id$ with its initial dependencies $D$ (line~\ref{epaxos:cmt:fast-cmt}).
  Since $\mathtt{Waiting}(\id', k')$ has been sent, there exists a recovery quorum $Q'$ and a set $R' \subseteq Q'$ such that $|R'| = k' > n-f-e$ and all processes in $R'$ pre-accepted $\id'$ with its initial dependencies $D'$ (line~\ref{epaxos:recvr:fast-cond}).
  Moreover, $\coord(\id')$ also pre-accepted $\id'$ with $D'$, and line~\ref{epaxos:recvr:coord-in} ensures that $\coord(\id') \not\in R'$.
  Since $n \geq 2e+f-1$, the intersection of $Q_f$ and $R' \cup \{\coord(\id')\}$ has a cardinality
\[
    \geq (n-e) + (n-f-e+2) - n = n-2e-f+2 \geq 1.
\]
Thus, some process pre-accepted $\id$ with $D$ and $\id'$ with $D'$.
From this we obtain a contradiction in the same way as in the proof of Lemma~\ref{prop:wait}.
\end{proof}

We next prove that, even though the optimized protocol restricts when the recovery can be aborted because of $\mathtt{Waiting}$ messages, the recovery eventually terminates.
\begin{lemma}
  \label{prop:all-good}
  Every command $\id$ is eventually committed at all correct processes.
\end{lemma}
\begin{proof}[Proof sketch]
  Assume by contradiction that some process never commits a command $\id$.
  Our coordinator nomination mechanism (\S\ref{sec:recover} and \tr{\ref{app:start-recovery}}{\nliveness}) ensures that a command keeps getting recovered until all correct processes commit it, and that eventually after $\GST$ and after failures have stopped, each command can only be recovered by a single process.
  Let $p$ be this process for the command $\id$, and let $D$ be $\id$'s initial dependency set.
  We can show that every recovery attempt by $p$ must get stuck at line~\ref{epaxos:recvr:wait}, meaning $\id$ always encounters a potentially invalidating command at line~\ref{epaxos:recvr:let-conflicts} that is itself blocked.

  As we consider runs with only finitely many commands, there is a command $\id'$ that blocks $\id$ infinitely many times, i.e., process $p$ computes a set $I$ at line~\ref{epaxos:recvr:rec-let-conflicts} such that $\id' \in I$.
  Eventually, after $\GST$, $\id'$ is also recovered only by a single process.
  Let $p'$ be this process, and $D'$ be the initial dependencies of $\id'$.
  Then, since $\id$ repeatedly blocks on $\id'$ (and in particular never satisfies the first disjunct at line~\ref{epaxos:recvr:conf-cmt}), line~\ref{epaxos:recvr:let-conflicts} implies:
  {\em (i)}~$\id' \notin D$;
  {\em (ii)}~$\id \notin D'$; and
  {\em (iii)}~the initial payloads of $\id$ and $\id'$ conflict.
  By line~\ref{epaxos:recvr:set-initcmd}, eventually, the initial payload of $\id$ is stored at a recovery quorum of $\id$.
  Recovery quorums of $\id$ and $\id'$ intersect, so eventually on every recovery of $\id'$, process $p'$ will receive a $\MValidateOK$ message from some correct process $q'$ storing $\id$'s initial payload.
  Notice that $q'$ never commits $\id$: otherwise $\id$'s recovery would eventually complete at line~\ref{epaxos:recvr:send-cmt} or \ref{epaxos:recvr:coord-commit}.
  Then, due to \emph{(i)}--\emph{(iii)}, process $q'$ will add $\id$ to the invalidating set at line~\ref{epaxos:recvr:let-conflicts}, so eventually, every recovery of $\id'$ by $p'$ will compute a set $I'$ at line~\ref{epaxos:recvr:rec-let-conflicts} such that $\id \in I'$.
  
  Due to line~\ref{epaxos:recvr:send-wait}, process $p'$ will eventually receive $\MWaiting(\id,k)$ from $p$.
  We must have $k \leq n-f-e$: otherwise, $p'$ would satisfy the condition at line~\ref{epaxos:recvr:break-cond-wait-r} and finish the recovery, contradicting our assumption.
  Hence, when $p$ recovers $\id$, it has $|\Rmax| \leq n - f - e$.
  On the other hand, by lines~\ref{epaxos:recvr:recoverok} and~\ref{epaxos:recvr:let-rmax},
  $|\Rmax| \ge |Q| - e \ge n - f - e \geq |\Rmax|$, so that $|\Rmax| = |Q| - e$.
  Since the conditions at lines~\ref{epaxos:recvr:coord-in} and~\ref{epaxos:recvr:break-cond-coord} never hold for $\id'$, $\coord(\id')$ must be faulty.
  But then eventually, $p$ will initiate a recovery of $\id$ for which the condition at line~\ref{epaxos:recvr:break-cond-wait} will hold, letting $p$ finish the recovery.
  This contradicts our assumption.
\end{proof}

The full proof of correctness for \epp is given in \tr{\ref{sec:proof}}{C} and uses the invariants and lemmas we have presented so far.
\begin{theorem}
  \label{theo:upper}
  The optimized {\em \epp} protocol implements an $f$-resilient $e$-fast SMR protocol, provided $n \ge \max\{2e+f-1, 2f+1\}$.
\end{theorem}

\section{Thrifty Protocol Version}
\label{sec:summary-thrifty}

The \epp protocol can also be changed to incorporate the {\em thrifty} optimization of the original \epaxos.
This allows the coordinator to take the fast path even when the rest of the fast quorum processes disagree with it on dependencies, provided they agree among themselves.
In exchange, the coordinator has to select the fast quorum it will use a priori, and communicate only with this quorum.
Hence, the coordinator will not be able to take the fast path even with a single failure, if this failure happens to be within the selected fast quorum.
Thus, the resulting protocol is no longer $e$-fast for any $e>0$ (\S\ref{sec:smr}).
We defer the description of this protocol to \tr{\ref{sec:thrifty}}{\nthrifty} and its proof of correctness to \tr{\ref{sec:proof}}{\nproof}.

\section{Comparison with the Original Egalitarian Paxos}
\labsection{comparison}

The main subtlety of the original \epaxos protocol for $e = \lceil\frac{f+1}{2}\rceil$ is the failure-recovery procedure.
When there is a suspicion that a command $\id$ took the fast path with payload $c$ and dependency set $D$ (akin to line~\ref{epaxos:recvr:fast-cond} in Figure~\ref{epaxos:recvr}), the new coordinator in \epaxos performs a {\em Tentative Pre-accept} phase: it tries to convince at least $f+1$ processes to pre-accept the values $c$ and $D$ for $\id$ (as in line~\ref{epaxos:cmt:receive-preaccept}).
Unlike our Validation phase, the Tentative Pre-accept phase has a major side effect: processes participating in it that did not already know about the command $\id$ pre-accept it with the dependencies suggested by the new coordinator.
As we now explain, this difference makes the protocol more complex and error-prone.

\subparagraph{Safety.}
As mentioned in \S\ref{sec:intro}, Egalitarian Paxos has a bug in its use of Paxos-like ballot variables: the protocol uses a single ballot variable instead of two.
Because recovery in this protocol is much more complex than in Paxos, fixing this is nontrivial.
In particular, when a process tentatively pre-accepts a command, it is unclear if the second variable (corresponding to our $\cbal$) should be updated or not.
We avoid this choice altogether by not forcing a process to commit to a choice of dependencies during the validation phase.

As we also mentioned in \S\ref{sec:intro}, the available protocol descriptions~\cite{epaxos,epaxos-thesis} only detail recovery for the thrifty version of the protocol, where the coordinator fixes a fast quorum for each command (\S\ref{sec:summary-thrifty}).
Adjusting it to the non-thrifty version (corresponding to our protocol in \S\ref{sec:base}-\ref{sec:coord-opt}) is nontrivial.
The reason is that, in \epaxos, the new coordinator recovering a command may encounter processes in its recovery quorum that tentatively pre-accepted dependencies proposed by past coordinators.
For each such process, the new coordinator needs to check whether an old coordinator is part of the initial fast path quorum~\cite[page 43, conditions 7.d and 7.e]{epaxos-thesis}.
It is unclear how this check can be done in the non-thrifty version, where {\em any} process can potentially be part of the fast quorum that committed the command.
Without this missing piece of information, the non-thrifty version of \epaxos is incomplete.

\subparagraph{Liveness.}
In the thrifty version of \epaxos, two recovery attempts for the same command with different dependencies can block each other indefinitely.
For example, some processes may tentatively pre-accept $D$ or $D'$, preventing either set from gathering the necessary $f+1$ pre-accepts and causing a recovery deadlock.
We detail this bug in \tr{\ref{sec:bug}}{\nbug}.

Another liveness problem comes from the fact that a process does not tentatively pre-accept $D$ and $c$ if it has already recorded an {\em interfering} command $\id'$ -- analogous to an invalidating command in \epp.
Any interfering command $\id'$ is reported back to the new coordinator.
When the coordinator receives this information, it pauses the recovery of $\id$ and starts recovering $\id'$.
As a consequence, multiple processes may simultaneously attempt to coordinate the same command, even after $\GST$, which may block progress.
In \epp we avoid this problem due to our use of $\mathtt{Waiting}$ messages (lines~\ref{epaxos:recvr:wait}-\ref{epaxos:recvr:break-cond-coord-end}).

\section{Related Work}
\labsection{related}

\subparagraph{State-machine replication.}
SMR is a classic approach to implementing fault-tolerant services~\cite{clocks,smr}.
Schiper~and~Pedone~\cite{gbroadcast}, along with Lamport~\cite{gpaxos}, observed that an SMR protocol can execute commutative commands in any order.
To determine the predecessors of a command, these works use partially ordered consensus instances.
Zielinski~\cite{optimisticGenericBroadcast} further observed that the partial order can be replaced by a directed graph, which may have cycles.
\epaxos~\cite{epaxos} is the first practical SMR protocol to use this approach, which decentralizes the task of ordering state-machine commands.

\subparagraph{Derivatives of \epaxos.}
The \epaxos protocol has inspired a number of follow-up works, surveyed in~\cite{losa16,sutra20,WhittakerGSHS21}. 
In particular, several proposals of protocols with similar architecture tried to achieve better performance in failure-free cases~\cite{atlas,caesar,tempo}.
However, none of these protocols has optimal fault-tolerance: the values of $e$ and $f$ they admit are lower than what is required by our Theorem~\ref{theo:lower}, and what \epaxos aimed to achieve.
Janus~\cite{shuai-osdi16} adapted \epaxos to implement transactions.
Due to the complexity of the original protocol, it used its unoptimized version where fast quorums contain {\em all} processes.
Our \epp protocol can serve as a basis for a variant of Janus that reduces the fast quorum size without compromising correctness.
Gryff~\cite{gryff} integrated \epaxos with ABD to improve the performance of writes in distributed key-value stores, while Tollman, Park and Ousterhout~\cite{epaxos-revisited} proposed an enhancement to \epaxos that improves its performance.
Since these works use \epaxos as a black box, they inherit its bugs.
Our version would be a drop-in correct replacement.

\subparagraph{Lower bounds.}
Pedone and Schiper~\cite{PedoneS04} have previously demonstrated that an SMR protocol can tolerate up to $e = f \leq \frac{n}{3}$ failures while executing a command in two message delays.
Lamport~\cite{lowerbounds} generalized this result, showing that any protocol requires $n \ge \max\{2e+f+1, 2f+1\}$ processes (matched by the protocol in \S\ref{sec:base}).
These prior works measure the minimal time it takes to execute a command at {\em all} correct processes.
In practice, a client typically only needs a fast response from the single replica to which it submitted the request.
Our recent work~\cite{BA} defined a notion of fast consensus using this more pragmatic approach.
In \refsection{smr}, we generalized this notion to SMR and established the corresponding lower bound on the number of processes: $n \ge \max\{2e+f-1, 2f+1\}$.
To the best of our knowledge, \epp is the first provably correct protocol to match this bound.

\clearpage

\bibliographystyle{plainurl}
\bibliography{epaxos_simplified}

\iflong
\appendix
\clearpage
\section{Proof of Theorem~\ref{theo:lower}}
\labappendix{lowerbound}

To prove Theorem~\ref{theo:lower}, we use the lower-bound result in \cite{BA} and the fact that (object) consensus reduces to SMR.
A consensus object is an atomic object with  a single operation, $\propose(v)$.
When a process invokes $\propose(v)$, it proposes the value $v$; the operation returns the consensus decision.
Consensus object ensures that {\em (i)} every decision is the proposal of some process; {\em (ii)} no two decisions are different; and that {\em (iii)} every correct process eventually decides.

We prove Theorem~\ref{theo:lower} by contradiction.
Suppose the existence of an $f$-resilient $e$-fast SMR protocol $\mathcal{P}$ satisfying $n < \max\{2e+f-1, 2f+1\}$.
To solve consensus with $\mathcal{P}$, we proceed as follows:
Commands are pairs $(v,p)$ where $v$ is some consensus input value and $p$ a process identifier.
Commands $(v,p)$ and $(v',p')$ conflict when $v \neq v'$.
To propose a value $v$, process $p$ creates command $(v,p)$ and submits it to $\mathcal{P}$.
The first command executed locally contains the consensus decision.

It is straightforward to see that the above construction implements object consensus:
No two processes execute conflicting commands in different orders.
Hence, they must take the same decision.
Moreover, by the Validity property of SMR, the decided value is proposed by some process.

Because protocol $\mathcal{P}$ tolerates up to $f$ crashes, the construction is an $f$-resilient solution to consensus.
We show that it is also $e$-two-step, as defined in \cite{BA}.
Recall that this precisely means that for all $E \subseteq \procSet$ of size $e$:
\begin{enumerate}[noitemsep,topsep=0pt,parsep=0pt]
\item For every value $v$ and process $p \in \procSet \setminus E$, there exists an $E$-faulty synchronous run in which only $p$ calls $\propose()$, the proposed value is $v$, and $p$ decides a value by time $2\Delta$.
\item For every value $v$ and process $p \in \procSet \setminus E$, there exists an $E$-faulty synchronous run in which all processes in $\procSet \setminus E$ call $\propose(v)$ at the beginning of the first round, and $p$ decides a value by time $2\Delta$.
\end{enumerate}

Fix a subset $E \subseteq \procSet$ of size $e$.
Assume some value $v$, a process $p \in \procSet \setminus E$, and a run of the above construction in which $p$ calls $\propose(v)$ at the start of the first round and the other processes do not propose anything.
Because $\mathcal{P}$ is $e$-fast, $p$ executes the command $(v,p)$, and thus decides, at time $2\Delta$.
Similarly, consider a run in which all the processes in $\procSet \setminus E$ propose the same value $v$ initially.
In this run, all the submitted commands are conflict-free.
Hence, each correct process decides by time $2\Delta$.

Theorem 2 in \cite{BA} states that an $f$-resilient $e$-two-step consensus object is implementable iff $n \ge \max\{2e + f - 1, 2f+1\}$.
This contradicts the existence of protocol $\mathcal{P}$.

\section{Thrifty Protocol Details}
\labsection{thrifty}

Figures~\ref{epaxos-thrifty:commit}--\ref{epaxos-thrifty:recvr} show the modifications to the optimized \epp (\S\ref{sec:coord-opt}) that yield the thrifty version of the protocol (for simplicity, we do not consider applying the thrifty mode to the baseline protocol of \S\ref{sec:base}).
Like in \S\ref{sec:coord-opt}, we assume $n \geq \max\{2e+f-1, 2f + 1\}$.

\subparagraph*{Modifications to the commit protocol (Figure~\ref{epaxos-thrifty:commit}).}
In the thrifty version, the initial coordinator sends the $\MPreAccept$ message to a single fast quorum, which must consist of exactly $n-e$ processes and includes the coordinator (line~\ref{epaxos:cmt:send-preacc-thrifty}t).
The fast path can be taken even with a dependency set $D$ different from the initial one, as long as all the processes in the fast quorum except the initial coordinator agree on it (line~\ref{epaxos:cmt:fast-cmt-thrifty}t).
An important restriction is that a command $\id$ cannot take the fast path if a fast quorum process detects that $\id$ must depend on a command in the $\startP$ phase (line~\ref{epaxos:thrifty:start-phase}t).
This situation can occur when, prior to receiving $\id$, a process in the fast quorum has received a $\MValidate$ message for a conflicting command $\id'$, stored its payload at line~\ref{epaxos:recvr:set-c}, but has not yet advanced $\id'$ out of the $\startP$ phase.
In this case the process adds an additional $\mathtt{notOK}$ argument to its $\MPreAcceptOK$ message (line~\ref{epaxos:thrifty:notOK}t), which disables the fast path (line~\ref{epaxos:cmt:fast-cmt-thrifty}t).
As we show in \S\ref{sec:proof}, this restriction is required to ensure liveness.
When a process sends a $\MPreAcceptOK$ message with $\mathtt{notOK}$, it moves $\id$ to the $\preacceptnotokP$ phase (line~\ref{epaxos:thrifty:notOK}t), and the coordinator of $\id$ processes a quorum of $\MPreAcceptOK$ messages only if it is in either the $\preacceptP$ or the $\preacceptnotokP$ phase (line~\ref{epaxos:cmt:preaccok-precond-thrifty}t).

\newcommand{\changed}[1]{{\SetNlSty{textbf}{}{t}#1}}
\newcommand{\unchanged}{\SetNlSty{}{}{}}
\begin{figure}[t!]
  \begin{minipage}{1.3\linewidth}
    \setcounter{AlgoLine}{7}
    \SetNlSty{}{}{}
      \scalebox{0.85}{
        \begin{algorithm}[H]
          \renewcommand{\;}{\\}
          \SetKwProg{Function}{}{:}{}
          \DontPrintSemicolon

          \Function{$\submit(c)$}{ %
            \Let $\id = \newid()$\;
            \changed{\send $\MPreAccept(\id, c, \{\id' \mid
            \inconflict{\cmd[\id']}{c}\})$ {\bf to \hl{any $Q$ such that $|Q| = n-e \land p \in Q$}}\;}
            \label{epaxos:cmt:send-preacc-thrifty}
          }

          \smallskip

          \SubAlgo{\onreceive $\MPreAccept(\id, c, D)$ \from $q$}{ %
            \precond $\bal[\id] = 0 \land \phase[\id] = \startP$\;
            $\assign{\cmd[\id]}{c}$\;
            $\assign{\initCmd[\id]}{c}$\;
            $\assign{\initDep[\id]}{D}$\;
            $\assign{\dep[\id]}{D \cup \{\id' \mid \inconflict{\cmd[\id']}{\cmd[\id]}\}}$\;
            \label{epaxos:thrifty:pre-accept-deps}
            \changed{\uIf{\hl{$\forall \id \in \dep[\id].\, \phase[\id] \neq \startP$}}{
              \label{epaxos:thrifty:start-phase}
              \nonl $\assign{\phase[\id]}{\preacceptP}$\;
              \nonl \send $\MPreAcceptOK(\id, \dep[\id], \text{\hl{$\mathtt{OK}$}})$ \To $q$\;
            }
              \Else{
                \label{epaxos:thrifty:notOK}
                \nonl $\assign{\phase[\id]}{\preacceptnotokP}$\;
              \nonl \send $\MPreAcceptOK(\id, \dep[\id], \text{\hl{$\mathtt{notOK}$}})$ \To $q$\;
            }}
          }

          \smallskip

          \SubAlgo{\changed{\onreceive $\MPreAcceptOK(\id, D_q, \text{\hl{$\mathit{OK}_q$}})$ {\bf from all $q \in Q$}}}{
            \label{epaxos:cmt:rec-preaccok-thrifty}
            \changed{\precond $\bal[\id] = 0 \land \text{\hl{$\phase[\id] \in \{\preacceptP, \preacceptnotokP\}$}} \land |Q| \geq n-f$}\;
            \label{epaxos:cmt:preaccok-precond-thrifty}
            \Let $D = \bigcup_{q \in Q} D_q$\;
            \changed{\uIf{$|Q| \geq n-e\ $\hl{$\land\ (\forall q \in Q\setminus\{p\}.\, D_q = D \land \mathit{OK}_q = \mathtt{OK})$}}{
              \label{epaxos:cmt:fast-cmt-thrifty}
              \unchanged \send $\MCommit(0, \id, \cmd[\id], D)$ \ToAll\;
            }}
            \Else{
              \send $\MAccept(0, \id, \cmd[\id], D)$ \ToAll\;
            }
          }
        \end{algorithm}
      }
  \end{minipage}

  \caption{Modifications to the commit protocol of the optimized \epp
    (\S\ref{sec:coord-opt}) that yield the thrifty version of the protocol: code
    at process $p$. Self-addressed messages are delivered immediately.}
  \label{epaxos-thrifty:commit}
\end{figure}

\begin{figure}[t!]
  \begin{minipage}{1.3\linewidth}
    \setcounter{AlgoLine}{56}
    \SetNlSty{}{}{}
    \scalebox{0.85}{
      \begin{algorithm}[H]
        \renewcommand{\;}{\\}
        \SetKwProg{Function}{}{:}{}
        \DontPrintSemicolon

        \SubAlgo{\nonl $\cdots$}{
        \SetAlgoNoEnd
        \SubAlgo{\changed{\textbf{else if} \hl{$\exists D.\, \exists R.\, |R|
              \geq |Q| - e \land (\forall q \in Q.\, \vphase_q = \preacceptP
              \Longleftrightarrow q \in R) \land (\forall q \in R.\, \vdep_q=
              D)$} \\ \nonl ~~\textbf{then}}}{ 
          \label{epaxos-thrifty:fastcond}
          {\bf let} $\Rmax$ be the largest set $R$ that satisfies the condition at line~\mbox{\ref{epaxos-thrifty:fastcond}t}\;
          \changed{\Let {\ $(c, D, \text{\hl{$D_0$}}) = (c_q, \vdep_q, \text{\hl{$\vinitDep_q$}})$ for any $q \in R$}}\;
          \label{epaxos-thrifty:let-d0}
          \changed{\send $\MValidate(b, \id, c, D, \text{\hl{$D_0$}})$ {\bf to all processes in $Q$}}\;
          \label{epaxos-thrifty:send-validate}
          \nonl\dots\;
        
        \smallskip
        \smallskip
        \setcounter{AlgoLine}{66}
        \Else{
          \changed{\send \hlc{$\mathtt{Waiting}(\id, \text{\hl{$D$}}, |\Rmax|)$} \ToAll\;}
          \label{epaxos-thrifty:sendwait}
          \SetAlgoNoLine
          \SetAlgoNoEnd
          \SubAlgo{\bf wait until}{
            \SetAlgoVlined
            \SetAlgoNoEnd
            \SubAlgo{{\bf case $\exists (\id', \_) \in \conflicts.\, \phase[\id'] = \commitP \land (\cmd[\id'] \neq \noop \land \id \notin \dep[\id'])$ do}}{
              \send $\MAccept(b, \id, \noop, \emptyset)$ \ToAll\;
            }
            \SubAlgo{{\bf case $\forall (\id', \_) \in \conflicts.\, \phase[\id'] = \commitP \land (\cmd[\id'] = \noop \lor \id \in \dep[\id'])$ do}}{
              \send $\MAccept(b, \id, c, D)$ \ToAll\;
            }
            \SubAlgo{\changed{{\bf case $\exists (\id', \_) \in \conflicts.\, (p$
                received $\MWaiting(\id', \text{\hl{$D'$}}, k')) \land k' >
                n-f-e\ \text{\hl{$\land\ \id \notin D'$}}$ do}}}{
              \label{epaxos-thrifty:break-two}
              \unchanged \send $\MAccept(b, \id, \noop, \emptyset)$ \ToAll\;
            }
            \SubAlgo{\changed{{\bf case $p$ received $\MRecoverOK(b,\id,\_,\vcmd,\vdep,\_,\vphase)$ from $q \notin Q$
                  with $\vphase = \commitP \lor \vphase = \acceptP \lor q = \coord(\id) $
                  \hspace{2cm} $\text{\hl{$\lor\ (\vphase =\preacceptP \land \vdep \neq D)$}}$ do}}}{
              \label{epaxos-thrifty:bad-phase}
              \lIf{$\vphase = \commitP$}{\send $\MCommit(b, \id, \vcmd, \vdep)$ \ToAll} %
              \lElseIf{$\vphase = \acceptP$}{\send $\MAccept(b, \id, \vcmd, \vdep)$ \ToAll} %
              \lElse{\send $\MAccept(b, \id, \noop, \emptyset)$ \ToAll} \label{epaxos-thrifty:bad}
            }
          }
        }
      }
        {\bf else} \send $\MAccept(b, \id, \noop, \emptyset)$ \ToAll
        \label{epaxos-thrifty:last-nop}
    }
    
        \smallskip
        \smallskip

        \SubAlgo{\changed{\onreceive $\MValidate(b, \id, c, D, \text{\hl{$D_0$}})$ \from $q$}}{
          \label{epaxos-thrifty:get-validate}
          \precond $\bal[\id] = b$\;
          $\assign{\cmd[\id]}{c}$\;
          $\assign{\initCmd[\id]}{c}$\;
          \changed{$\assign{\initDep[\id]}{\text{\hl{$D_0$}}}$}\;
          \label{epaxos-thrifty:set-d0}
          \nonl\dots\;
          \smallskip
        }
      \end{algorithm}
    }
  \end{minipage}
  \caption{Modifications to the recovery protocol of the optimized \epp
    (\S\ref{sec:coord-opt}) that yield the thrifty version of the protocol: code
    at process $p$. Self-addressed messages are delivered immediately.}
  \label{epaxos-thrifty:recvr}
\end{figure}

\subparagraph*{Modifications to the recovery protocol (Figure~\ref{epaxos-thrifty:recvr}).}
Since in the thrifty version the fast path can be taken with any dependency set, we also need to adjust the recovery protocol:
\begin{itemize}
\item
  In the thrifty version, a new coordinator suspects that the command $\id$ may have taken a fast path only if all processes in the recovery quorum pre-accepted $\id$ with the same dependencies $D$, and there are at least $|Q|-e$ such processes (line~\ref{epaxos-thrifty:fastcond}t).
  Since the thrifty version uses a single fast quorum, any disagreement on pre-accepted dependencies at line~\ref{epaxos-thrifty:fastcond}t indicates that the command could not have taken the fast path.
\item
  Each $\MValidate$ message now includes an additional argument to carry the initial dependency set (lines~\ref{epaxos-thrifty:let-d0}t--\ref{epaxos-thrifty:send-validate}t, \ref{epaxos-thrifty:get-validate}t, \ref{epaxos-thrifty:set-d0}t): in the thrifty version this set may be different from the set $D$ used in recovery (line~\ref{epaxos-thrifty:let-d0}t).
\item
  In the thrifty version we impose an extra constraint on when the new coordinator can abort a command $\id$ if it receives a $\MWaiting$ message for a potentially invalidating command (line~\ref{epaxos:recvr:break-cond-wait-r} in Figure~\ref{epaxos:recvr}).
  This is because the proof of Lemma~\ref{prop:wait-opt} (\S\ref{sec:coord-opt}) assumes that the fast path can only be taken with the initial dependencies, which is no longer the case in the thrifty version.
  To deal with this, $\MWaiting$ messages now also carry the suspected fast path dependencies of the command being recovered (line~\ref{epaxos-thrifty:sendwait}t).
  The new coordinator of a command $\id$ aborts this command only if it does not belong to the fast path dependencies of the potentially invalidating command (line~\ref{epaxos-thrifty:break-two}t).
\item
  Finally, to ensure liveness, we add an extra case in which a new coordinator recovering $\id$ with a dependency set $D$ can abort the recovery when it receives an additional $\MRecoverOK$ message from outside of the original recovery quorum (line~\ref{epaxos:recvr:break-cond-coord} in Figure~\ref{epaxos:recvr}).
  Namely, the new coordinator can now abort the recovery of $\id$ if it receives $\MRecoverOK(\_, \id, \_, \_, \vdep, \_, \preacceptP)$ with $\vdep \neq D$ (line~\ref{epaxos-thrifty:bad-phase}t): in this case $\id$ could not have taken the fast path with $D$, since some process in the fast quorum pre-accepted a different dependency set.
\end{itemize}

\begin{theorem}
  \label{theo:upper-thrifty}
  The thrifty version of {\em \epp} protocol implements an $f$-resilient $0$-fast SMR protocol, provided $n \ge \max\{2e+f-1, 2f+1\}$.
\end{theorem}

Note that, even though the thrifty protocol is only $0$-fast, the parameter $e$ still matters, as it determines the size of fast quorums.

\section{Starting Recovery}
\label{app:start-recovery}

\begin{figure}
  \begin{minipage}{\linewidth}
    \scalebox{0.9}{
      \begin{algorithm}[H]
        \setcounter{AlgoLine}{87}
        \renewcommand{\;}{\\}
        \SetKwProg{Function}{function}{:}{}
        \DontPrintSemicolon

        \SubAlgo{{\bf periodically (with increasing delays)}}{
          \label{epaxos:rec-start:start}

          \uIf{$\phase[\id] \neq \commitP$}{
            \label{epaxos:rec-start:check-commit}
            \send $\MTryRecover(\id)$ \To{} $\Omega[\id]$\;
            \label{epaxos:rec-start:send-try}
          }
          \ElseIf{$\cmd[\id] = \noop \land p = \coord(\id)$}{
            \label{epaxos:rec-start:check-nop}
            $\submit(\initCmd[\id])$\;
            \label{epaxos:rec-start:submit}
          }
        }

        \smallskip

        \SubAlgo{\onreceive $\MTryRecover(\id)$}{
          \lIf{$p = \Omega[\id]$}{$\recover(\id)$}
          \label{epaxos:rec-start:rec}
        }

      \end{algorithm}
    }
  \end{minipage}
  \caption{\epp: recovery policy of a command $\id$ at a replica $p$.}
  \label{epaxos:rec-start}
\end{figure}

In Figure~\ref{epaxos:rec-start}, we present a protocol that a replica $p$ uses to invoke the recovery for a command $\id$.
The mechanism relies on per-command leader detectors: each command $\id$ has an associated leader detector $\Omega[\id]$, which eventually outputs the same correct process at all replicas~\cite{omega}.
Replica $p$ periodically checks whether $\id$ is committed.
If not, it sends a $\MTryRecover$ message to $\Omega[\id]$ (line~\ref{epaxos:rec-start:send-try}), asking this to become the coordinator of $\id$.
Upon receiving $\MTryRecover(\id)$, a replica checks whether it considers itself the coordinator of $\id$ and, if so, starts the recovery (line~\ref{epaxos:rec-start:rec}).

Recall that in \epp, a command $\id$ submitted with a payload $c$ may sometimes be committed as a $\noop$ as a result of recovery.
To ensure that every command submitted by a correct process is eventually executed, a process resubmits its original payload if it detects that $\id$ was replaced by a $\noop$ (lines~\ref{epaxos:rec-start:check-nop}--\ref{epaxos:rec-start:submit}).

\section{Correctness of \epp}
\label{sec:proof}

\subsection{Proof of Invariants~\ref{inv:agr} (Agreement) and~\ref{inv:cnst} (Visibility)}

This proof covers all the variants of the protocol we presented: the baseline protocol (\S\ref{sec:base}), its optimized version (\S\ref{sec:coord-opt}), and its thrifty version (\S\ref{sec:thrifty}).
In the following we make it clear which parts of the proof are relevant for which version of the protocol.
The proof relies on the low-level invariants listed in Figure~\ref{fig:invs}.
We omit the proofs of Invariants~\ref{inv:pre-acc-once}--\ref{inv:cmd-or-noop}, as they easily follow from the structure of the protocol.

\begin{figure*}[t!]
  \begin{enumerate}[label={\arabic*.}, ref={\arabic*}]
    \setcounter{enumi}{2}
  \item
    Each command is pre-accepted at most once by each process.
    \label{inv:pre-acc-once}
  \item
    Once a process moves a command to the $\acceptP$ or $\commitP$ phase, it remains in one of these phases for the entire execution.
    \label{inv:incr-phases}
  \item
    At every process and for every $\id$, if $\cbal[\id] > 0$, then $\phase[\id] \in \{\acceptP, \commitP\}$.
    \label{inv:bmax-acc}
  \item
    Assume that at time $t$ a process $p$ pre-accepts or validates a command $\id$ with a payload $c$ (i.e., $p$ executes line~\ref{epaxos:cmt:set-c} or line~\ref{epaxos:recvr:set-c}).
    At any moment after $t$, the process $p$ has $\cmd[\id] = c$ or $\cmd[\id] = \noop$, with the latter possible only if $\phase[\id] \neq \preacceptP$.
    \label{inv:cmd-or-noop}
  \item
    Assume that a process sends $\MCommit(0, \id, c, D)$ on the fast path (line~\ref{epaxos:cmt:send-cmt}).
    If a process executes line~\ref{epaxos:recvr:fast-cond} (or line~\ref{epaxos-thrifty:fastcond}t), then the condition in this line evaluates to true, and only for the dependency set $D$.
    \label{inv:eq-dep}
  \item
    If a process sends $\MPreAcceptOK(\id,D[,\_])$ and $\MPreAcceptOK(\id',D'[,\_])$ for
    two commands with conflicting initial payloads, then $\id \in D'$ or $\id' \in D$.
    \label{inv:two-preaccs}
  \item
    Assume that a quorum of processes has received $\MAccept(b, \id, c, D)$ and replied with $\MAcceptOK(b, \id)$.
    \begin{enumerate}
    \item For any $\MAccept(b', \id, c', D')$ sent, if $b' \geq b$, then $c' = c$ and $D' = D$;
      \label{inv:agr-acc}
    \item For any $\MCommit(b', \id, c', D')$ sent, if $b' \geq b$, then $c' = c$ and $D' = D$.
      \label{inv:agr-cmt}
    \end{enumerate}
    \label{inv:agr-aux}
  \item
    Assume that a process sends $\MCommit(0, \id, c, D)$ at line~\ref{epaxos:cmt:send-cmt}.
    \begin{enumerate}
    \item For any $\MAccept(\_, \id, c', D')$ sent, $c' = c$ and $D' = D$;
      \label{inv:cmt-acc}
    \item For any $\MCommit(\_, \id, c', D')$ sent, $c' = c$ and $D' = D$.
      \label{inv:cmt-cmt}
    \end{enumerate}
    \label{inv:cmt-aux}
  \end{enumerate}
  \caption{Additional invariants of \epp.
    All invariants hold for all three versions of \epp.
  }
  \label{fig:invs}
\end{figure*}

\begin{proof}[Proof of Invariant~\ref{inv:eq-dep} (all protocol versions)]
  Assume that $\id$ takes the fast path, i.e., a process sends $\MCommit(0, \id, \_, D)$ from line~\ref{epaxos:cmt:send-cmt}.
  By line~\ref{epaxos:cmt:fast-cmt} (line~\ref{epaxos:cmt:fast-cmt-thrifty}t in the thrifty version), there exists a quorum $Q_f$ of size $\geq n-e$ whose processes pre-accepted $\id$.
  Consider a process that reaches line~\ref{epaxos:recvr:fast-cond} (line~\ref{epaxos-thrifty:fastcond}t in the thrifty version) after receiving $\MRecoverOK(\_, \id, \_, \_, \vdep_p, \_, \vphase_p)$ messages from each process $p$ in a quorum $Q$.
  Since the conditions at lines~\ref{epaxos:recvr:is-cmt}--\ref{epaxos:recvr:coord-in} do not hold, for every $p \in Q \cap Q_f$, $\vphase_p = \preacceptP$
  Since $|Q_f| \geq n - e$, we have $|Q \cap Q_f| \geq |Q| - e$.
  Thus, the first conjunct of line~\ref{epaxos:recvr:fast-cond} (line~\ref{epaxos-thrifty:fastcond}t in the thrifty version) is satisfied for $R = Q \cap Q_f$.

  We now consider the other conjuncts.
  For the non-thrifty versions (line~\ref{epaxos:recvr:fast-cond}), the command $\id$ can only be pre-accepted with the initial dependency set. Thus, the condition at line~\ref{epaxos:recvr:fast-cond} holds for $R = Q \cap Q_f$.

  For the thrifty version (line~\ref{epaxos-thrifty:fastcond}t), observe that any process that ever pre-accepted $\id$ must belong to $Q_f$ (line~\ref{epaxos:cmt:send-preacc-thrifty}).
  Moreover, $\coord(\id)$ is the only process that could have pre-accepted $\id$ with a dependency set different from $D$.
  However, by line~\ref{epaxos:recvr:coord-in}, $\coord(\id) \notin Q$, and therefore $\coord(\id) \notin R$.
  Hence, every process in $Q$ that pre-accepted $\id$ did so with the dependency set $D$, and all processes in $Q \cap Q_f$ pre-accepted $\id$.
  Thus, the condition at line~\ref{epaxos-thrifty:fastcond}t holds for the above-chosen $D$ and only for this $D$.
\end{proof}

\begin{proof}[Proof of Invariant~\ref{inv:two-preaccs} (all protocol versions)]
  Assume that a process $p$ pre-accepts a command $\id$ with an initial payload $c$ and a dependency set $D$, as well as a command $\id'$ with a conflicting initial payload $c'$ and a dependency set $D'$.
  Without loss of generality, assume that the former happens before the latter.
  By the time $p$ sends $\MPreAcceptOK(\id',D'[,\_])$ it has $\cmd[\id] = c$ or $\cmd[\id] = \noop$ (Invariant~\ref{inv:cmd-or-noop}).
  Then, since $\inconflict{c}{c'}$ and $\inconflict{\noop}{c'}$, by line~\ref{epaxos:cmt:preacc-deps}, $\id \in D'$, as required.
\end{proof}

We have already proved Lemma~\ref{prop:wait} in \S\ref{sec:recover} and Lemma~\ref{prop:red-one} for the optimized protocol in \S\ref{sec:coord-opt}.
We now prove Lemma~\ref{prop:red-one} for the thrifty protocol.

\begin{proof}[Proof of Lemma~\ref{prop:red-one} (thrifty protocol)]
  By contradiction, assume that $\id$ takes the fast path with a fast quorum $Q_f$.
  Hence, $|Q_f| \geq n-e$, all processes in $Q_f$ except $\coord(\id)$ pre-accepted $\id$ with a dependency set $D$ (line~\ref{epaxos:cmt:fast-cmt-thrifty}t), and $\coord(\id)$ pre-accepted $\id$ with a set $D_0 \subseteq D$ (because it handles its own $\MPreAccept$ message immediately).
  By our assumption, $|\Rmax| = |Q|-e$ and for some $\id'$ and $\vphase'$ we have $(\id', \vphase') \in \conflicts$, $\coord(\id') \notin Q$, and $\vphase' \neq \commitP$.
  Before executing line~\ref{epaxos:recvr:conf-cmt}, a recovering process executes line~\ref{epaxos-thrifty:fastcond}t, and satisfies the condition at this line for $D$ (Invariant~\ref{inv:eq-dep}).
  Hence, $\id' \notin D$ (line~\ref{epaxos:recvr:let-conflicts}).
  Let $D'$ be the initial dependency set of $\id'$.
  Notice that by lines~\ref{epaxos:cmt:set-d0} and \ref{epaxos-thrifty:set-d0}t, whenever a process has $\initCmd[\id'] \neq \bot$, it also has $\initDep[\id'] = D'$.
  Then by line~\ref{epaxos:recvr:let-conflicts}, the process whose $\MValidateOK$ message resulted in
  $(\id', \vphase')$ being included into $I$ at line~\ref{epaxos:recvr:rec-let-conflicts} must have had $\id \notin \initDep[\id'] = D'$.
  Moreover, by the same line, $\id$ and $\id'$ have conflicting initial payloads.
  Since $|\Rmax| = |Q| - e$, exactly $|Q| - e$ processes in $Q$ pre-accepted $\id$ with dependency set $D$, while the remaining $e$ processes did not pre-accept $\id$ at all (line~\ref{epaxos-thrifty:fastcond}t).
  Since $|Q_f| \geq n - e$ processes pre-accepted $\id$, all processes that are not in $Q$ must belong to $Q_f$.
  Then since $\coord(\id') \notin Q$, we must have $\coord(\id') \in Q_f$.
  Recall that all the processes in $Q_f$ pre-accepted $\id$ with either a dependency set $D$ or $D_0 \subseteq D$ (line~\ref{epaxos:cmt:fast-cmt-thrifty}t).
  In particular, this is the case for $\coord(\id')$.
  Moreover, $\coord(\id')$ also pre-accepted $\id'$ with $D'$.
  Since the initial payloads of $\id$ and $\id'$ conflict, by Invariant~\ref{inv:two-preaccs}, either $\id' \in D$ or $\id \in D'$: a contradiction.
\end{proof}

\begin{lemma}[thrifty protocol, couterpart of Lemma~\ref{prop:wait-opt}]
  \label{prop:wait-opt-plus}
  Assume that $n \geq 2e + f - 1$ and that $\mathtt{Waiting}(\id, D, k)$ and $\mathtt{Waiting}(\id', D', k')$ have been sent for commands with conflicting initial payloads such that $\id \notin D'$ and $\id' \notin D$.
  If $k' > n - f - e$, then $\id$ did not take the fast path.
\end{lemma}
\begin{proof}
  By contradiction, assume that $\id$ takes the fast path with a fast quorum $Q_f$.
  Hence, $|Q_f| \geq n-e$, all processes in $Q_f$ except $\coord(\id)$ pre-accepted $\id$ with a dependency set $\hat{D}$ (line~\ref{epaxos:cmt:fast-cmt-thrifty}t), and $\coord(\id)$ pre-accepted $\id$ with a set $D_0 \subseteq \hat{D}$ (because it handles its own $\MPreAccept$ message immediately).
  By Invariant~\ref{inv:eq-dep}, the condition at line~\ref{epaxos-thrifty:fastcond}t is satisfied for $\hat{D}$.
  Since $\mathtt{Waiting}(\id, D, k)$ has been sent, it must be the case that $\hat{D} = D$.
  Since $\mathtt{Waiting}(\id', D', k')$ has been sent, there exists a recovery quorum $Q'$ and a set $R' \subseteq Q'$ such that $|R'| = k' > n-f-e$ and all processes in $R'$ pre-accepted $\id'$ with dependencies $D'$ (line~\ref{epaxos-thrifty:fastcond}t).
  Moreover, $\coord(\id')$ pre-accepted $\id'$ with a set $D_0' \subseteq D'$, and line~\ref{epaxos:recvr:coord-in} ensures that $\coord(\id') \notin R'$.
  Since $n \geq 2e+f-1$, the intersection of $Q_f$ and $R' \cup \{\coord(\id')\}$ has a cardinality
\[
  \geq (n-e) + (n-f-e+2) - n = n-2e-f+2 \geq 1.
\]
  Thus, some process pre-accepted $\id$ with either $D$ or $D_0 \subseteq D$, and this process also pre-accepted $\id'$ with either $D'$ or $D'_0 \subseteq D'$.
  But then by Invariant~\ref{inv:two-preaccs}, either $\id' \in D$ or $\id \in D'$, contradicting our assumption.
\end{proof}

\begin{proof}[Proof of Invariant~\ref{inv:agr-acc} (all protocol versions)]
  Suppose a quorum of processes $Q_0$ has received $\MAccept(b, \id, c, D)$ and replied with $\MAcceptOK(b, \id)$.
  We prove by induction on $b'$ that for any $\MAccept(b', \id, c', D')$ sent, if $b' \geq b$, then $c' = c$ and $D' = D$;

  \emph{Base case ($b' = b$).}
  This is immediate because $\MAccept$ messages are sent at most once per ballot, and each ballot is owned by a single process.

  \emph{Induction step ($b'>b$).}
  We only consider the handler at line~\ref{epaxos:recvr:recoverok}, since it is the only handler that causes processes to send $\MAccept$ messages for ballots greater than $0$.
  Assume that a process $p$ executes this handler upon receiving a $\MRecoverOK(b', \id, \vcbal_q, c_q, \vdep_q, \vinitDep_q, \vphase_q)$ message from every process $q$ in a quorum $Q$.
  Let $\bmax = {\tt max}\{\vcbal_q \mid q \in Q\}$ be the ballot computed at line~\ref{epaxos:recvr:bmax}.
  Since $Q_0 \cap Q \neq \emptyset$, some process $q$ sends an $\MAcceptOK(b, \id)$ message and a $\MRecoverOK$ message at ballots $b$ and $b'$, respectively.
  When $q$ sends $\MRecoverOK$ it sets $\bal[\id]$ to $b'$.
  Since $b' > b$, by the precondition at line~\ref{epaxos:cmt:accept-precond}, $q$ must send $\MAcceptOK(b, \id)$ before sending $\MRecoverOK$.
  Therefore, $b \leq \bmax < b'$.

  Assume first that $\bmax = b$.
  After sending $\MAcceptOK$ at ballot $b$, process $q$ has $\phase[id] \in \{\acceptP, \commitP\}$ forever (Invariant~\ref{inv:incr-phases}).
  Hence, $q$ satisfies either the condition at line~\ref{epaxos:recvr:is-cmt} or the one at line~\ref{epaxos:recvr:is-acc}.
  Thus, process $p$ sends $\MAccept(b', \id, c', D')$ from line~\ref{epaxos:recvr:send-acc}.
  Consider any $q' \in Q$ satisfying the condition at line~\ref{epaxos:recvr:is-acc}.
  Since $\vcbal_{q'} = \bmax = b$, process $q'$ stores the payload and the dependencies that it received in the $\MAccept(b, \id, c, D)$ message from the coordinator of $\id$ at $b$.
  Hence, $c' = c$ and $D' = D$, as required.

  Assume now that $\bmax > b$; then $\bmax > 0$.
  Furthermore, by Invariant~\ref{inv:bmax-acc}, at any process, any command with a ballot greater than $0$ is either accepted or committed.
  This means that process $p$ sends $\MAccept(b', \id, c', D')$ from line~\ref{epaxos:recvr:send-acc}.
  By line~\ref{epaxos:recvr:is-acc}, there exists a process $q' \in Q$ with $\vphase_{q'} = \acceptP$, $\vcbal_{q'} = \bmax$, $c_{q'} = c'$ and $\vdep_{q'} = D'$.
  The process $q'$ thus received $\MAccept(\bmax, \id, c', D')$, so by induction hypothesis, $c' = c$ and $D' = D$, as required.
\end{proof}

\begin{proof}[Proof of Invariant~\ref{inv:agr-cmt} (all protocol versions)]
  Suppose a quorum of processes has received $\MAccept(b, \id, c, D)$ and replied with $\MAcceptOK(b, \id)$.
  We prove that for any $\MCommit(b', \id, c', D')$ sent, if $b' \geq b$, then $c' = c$ and $D' = D$.
  The proof is by induction on $b'$.

  Let $p$ be the process that sends $\MCommit(b', \id, c', D')$.
  Then this must happen at one of the following lines:~\ref{epaxos:cmt:send-cmt},~\ref{epaxos:cmt:sent-cmt-slow},~\ref{epaxos:recvr:is-cmt}, or~\ref{epaxos:recvr:coord-commit}.

  \emph{Base case ($b' = b$).}
  In this case $p$ is the process that sends the $\MAccept(b,\id,c,D)$ message.
  We make a case split on the line at which $p$ sends the $\MCommit$ message:
  \begin{itemize}
  \item \emph{Line~\ref{epaxos:cmt:send-cmt}.} By the precondition in line~\ref{epaxos:cmt:preaccok-precond}, $b'=b=0$.
    But then $p$ must execute line~\ref{epaxos:cmt:sent-acc} and cannot execute line~\ref{epaxos:cmt:send-cmt}: a contradiction.
  \item \emph{Line~\ref{epaxos:cmt:sent-cmt-slow}.}
  Because $p$ immediately receives any self-addressed $\MAccept$ message, this process must have $\cmd[\id]=c$ and $\dep[\id]=D$ when it executes line~\ref{epaxos:cmt:sent-cmt-slow}, as required.
  \item \emph{Line~\ref{epaxos:recvr:is-cmt} or~\ref{epaxos:recvr:coord-commit}.}
    In this case $p$ never sends the $\MAccept(b,\id,c,D)$ during the execution: a contradiction.
  \end{itemize}

  \emph{Induction step ($b'>b$).}
  We again make a case split on the line at which $p$ sends the $\MCommit$ message:
  \begin{itemize}
  \item \emph{Line~\ref{epaxos:cmt:send-cmt}.}
  Due to the precondition in line~\ref{epaxos:cmt:preaccok-precond}, $b'=b=0$: a contradiction.
  \item \emph{Line~\ref{epaxos:cmt:sent-cmt-slow}.}
  As the precondition in line~\ref{epaxos:cmt:receive-acceptok} holds, $p$ received a quorum of $\MAcceptOK(b',\id)$ messages.
  It follows that $p$ had previously sent an $\MAccept(b',\id,c',D')$ message.
  Then by Invariant~\ref{inv:agr-acc}, $c'=c$ and $D'=D$, as required.
  \item \emph{Line~\ref{epaxos:recvr:is-cmt} or~\ref{epaxos:recvr:coord-commit}.}
    The required follows immediately from the induction hypothesis.
\end{itemize}
\end{proof}

\begin{proposition}[all protocol versions]
  For every process $p$ and every command identifier $\id$, at any time when $p$ has $\cmd[\id] = c \neq \noop$, $\phase[\id] = \acceptP$ and $\dep[\id] = D$, one of the following two properties holds:
  \begin{enumerate}
  \item there exists a quorum $Q$ such that every process $q \in Q$ previously sent a $\MPreAcceptOK(\id,D_q[,\_])$ message with $\bigcup_{q \in Q} D_q = D$; or
    \label{prop:2dep-from-q}
  \item a process previously sent an $\MAccept(\_, \id, c, D)$ message from line~\ref{epaxos:recvr:send-acc-good} or line~\ref{epaxos:recvr:nop-2}.
    \label{prop:2dep-from-rec}
  \end{enumerate}
  \label{prop:2dep-from}
\end{proposition}
\begin{proof}
  By induction on the length of the execution.
  It is easy to see that the proposition holds initially.
  For the induction step, suppose a process $p$ either:
  \emph{(i)} sets $\phase[id]$ to $\acceptP$ and $\cmd[id]$ to $c \neq \noop$; or
  \emph{(ii)} sets $\cmd[\id]$ to $c \neq \noop$ while already having $\phase[id] = \acceptP$ and does not change $\phase[id]$.
  We examine each of these cases in turn.

  Case {\em (i)} may only happen at line~\ref{epaxos:cmt:set-acceptp}, after $p$ receives an $\MAccept(\_,\id,c,D)$ message.
  Since $c \neq \noop$, this message is sent by some process $p'$ at one of the following lines: \ref{epaxos:cmt:sent-acc}, \ref{epaxos:recvr:send-acc}, \ref{epaxos:recvr:send-acc-good}, \ref{epaxos:recvr:nop-2}, or \ref{epaxos:recvr:coord-acc}.
  We consider each of these cases separately:
  \begin{itemize}
  \item \emph{Line~\ref{epaxos:cmt:sent-acc}.} In this case $p'$ received a $\MPreAcceptOK(\id,D_q[,\_])$ message from some quorum $Q$ at line~\ref{epaxos:cmt:rec-preaccok}, with $\bigcup_{q \in Q} D_q = D$, as required.
  \item \emph{Line~\ref{epaxos:recvr:send-acc-good} or~\ref{epaxos:recvr:nop-2}.} The required follows immediately.
  \item \emph{Line~\ref{epaxos:recvr:send-acc} or~\ref{epaxos:recvr:coord-acc}.} The required follows by the induction hypothesis.
\end{itemize}

  Next, let us examine case {\em (ii)}, which may only happen at line~\ref{epaxos:recvr:set-c}.
  Since $p$ already accepted $\id$ when executing this line, it must have reported $\id$ as accepted in its last $\MRecoverOK$ response.
  Then the coordinator of $\id$ evaluated the condition in line~\ref{epaxos:recvr:is-acc} to true and never sends a $\MValidate$ message that would allow $p$ to execute line~\ref{epaxos:recvr:set-c}.
  Hence, case {\em (ii)} is impossible.
\end{proof}

\begin{proposition}[all protocol versions]
  For every process $p$ and every command identifier $\id$, at any time when $p$ has $\cmd[\id] = c \neq \noop$, $\phase[\id] = \commitP$ and $\dep[\id] = D$, one of the following two properties holds:
  \begin{enumerate}
  \item there exists a quorum $Q$ such that every process $q \in Q$ previously sent a $\MPreAcceptOK(\id,D_q[,\_])$ message with $\bigcup_{q \in Q} D_q = D$; or
    \label{prop:dep-from-q}
  \item a process previously sent an $\MAccept(\_, \id, c, D)$ message from line~\ref{epaxos:recvr:send-acc-good} or line~\ref{epaxos:recvr:nop-2}.
    \label{prop:dep-from-rec}
  \end{enumerate}
  \label{prop:dep-from}
\end{proposition}
\begin{proof}
  By induction on the length of the execution.
  It is easy to see that the proposition holds initially.
  For the induction step, suppose a process $p$ either:
  \emph{(i)} sets $\phase[id]$ to $\commitP$ and $\cmd[id]$ to $c \neq \noop$; or
  \emph{(ii)} sets $\cmd[\id]$ to $c \neq \noop$ while already having $\phase[id] = \commitP$ and does not change $\phase[id]$.
  We examine each of these cases in turn.

  Case {\em (i)} may only happen at line~\ref{epaxos:cmt:set-phase}, after $p$ receives a $\MCommit(\_,\id,c,D)$ message.
  Since $c \neq \noop$, this message is sent by some process $p'$ at one of the following lines: \ref{epaxos:cmt:send-cmt},~\ref{epaxos:cmt:sent-cmt-slow},~\ref{epaxos:recvr:send-cmt}, or~\ref{epaxos:recvr:coord-commit}.
  We consider each of these cases separately:
  \begin{itemize}
  \item \emph{Line~\ref{epaxos:cmt:sent-acc}.} In this case $p'$ received a $\MPreAcceptOK(\id,D_q[,\_])$ message from some quorum $Q$ at line~\ref{epaxos:cmt:rec-preaccok}, with $\bigcup_{q \in Q} D_q = D$, as required.
  \item \emph{Line~\ref{epaxos:cmt:sent-cmt-slow}.} In this case $p'$ received a quorum of $\MAcceptOK$ messages, which were sent in reply to a $\MAccept(\_,\id,c,D)$ message that $p'$ had previously sent. The required then follows from Proposition~\ref{prop:2dep-from}.
  \item \emph{Line~\ref{epaxos:recvr:send-cmt} or~\ref{epaxos:recvr:coord-commit}.} The required follows by the induction hypothesis.
\end{itemize}

  Next, let us examine case {\em (ii)}, which may only happen at line~\ref{epaxos:recvr:set-c}.
  Since $p$ already committed $\id$ when executing this line, it must have reported $\id$ as committed in its last $\MRecoverOK$ response.
  Then the coordinator of $\id$ evaluated the condition in line~\ref{epaxos:recvr:is-cmt} to true and never sends a $\MValidate$ message that would allow $p$ to execute line~\ref{epaxos:recvr:set-c}.
  Hence, case {\em (ii)} is impossible.
\end{proof}

\begin{proof}[Proof of Invariant~\ref{inv:cmt-aux} (all protocol versions)]
  Assume that a process sends $\MCommit(0, \id, c, D)$ at line~\ref{epaxos:cmt:send-cmt}, i.e., $\id$ takes the fast path with payload $c$ and dependency set $D$.
  Let $Q_f$ be the fast quorum used for this.
  We prove by induction on the length of the execution that:
  \begin{enumerate}[a.]
  \item For any $\MAccept(b, \id, c', D')$ sent by any process $p$, $c' = c$ and $D' = D$; and
  \item For any $\MCommit(b, \id, c', D')$ sent by any process $p$, $c' = c$ and $D' = D$.
  \end{enumerate}
  These properties trivially hold at the beginning of the execution.
  For the induction step, we consider the two cases separately.

  \emph{Case (a).} Observe that, if $b=0$, then the $\MAccept(b, \id, c', D')$ message must be sent at line~\ref{epaxos:cmt:sent-acc}, i.e., on the slow path, which is impossible as $\id$ takes the fast path.

  Now, consider the handler at line~\ref{epaxos:recvr:recoverok}, which is the only handler that results in process $p$ sending the $\MAccept$ message with $b>0$.
  Let $Q$ be a quorum of $\MRecoverOK(b, \id, \vcbal_q, c_q, \vdep_q, \vinitDep_q, \vphase_q)$ messages with which $p$ executes the handler at line~\ref{epaxos:recvr:recoverok}, and let $\bmax = {\tt max}\{\vcbal_q \mid q \in Q\}$ be the ballot the process computes at line~\ref{epaxos:recvr:bmax}.
  Process $p$ sends $\MAccept(b, \id, c', D')$ from one of the following lines: \ref{epaxos:recvr:send-acc}, \ref{epaxos:recvr:send-acc-coord}, \ref{epaxos:recvr:send-acc-good}, \ref{epaxos:recvr:conf-cmt-acc}, \ref{epaxos:recvr:nop-1}, \ref{epaxos:recvr:nop-2}, \ref{epaxos:recvr:break-cond-wait-r-end}, \ref{epaxos:recvr:coord-acc}, \ref{epaxos:recvr:break-cond-coord-end}, or \ref{epaxos:recvr:last-acc}.
  We now make a case split on these options.
  \begin{itemize}
  \item
    {\em Lines~\ref{epaxos:recvr:send-acc} (all protocol versions), \ref{epaxos:recvr:coord-acc} (optimized and thrifty protocols).}
    In all these cases process $p$ sends the $\MAccept$ message with the payload $c'$ and dependency set $D'$ received from a process $q$ that previously received $\MAccept(b', \id, c', D')$ message.
    But then by the induction hypothesis for case (a), $c = c'$ and $D = D'$, as required.
  \item
    {\em Line~\ref{epaxos:recvr:send-acc-coord} (optimized and thrifty protocols).}
    Since the conditions at lines~\ref{epaxos:recvr:is-cmt} and \ref{epaxos:recvr:is-acc} do not hold, by Invariant~\ref{inv:bmax-acc} we have $\bmax = 0$.
    Hence, the coordinator $\coord(\id) \in Q$ sends a message of the form $\MRecoverOK(b, \id, 0, c, \_, \_, \vphase)$.
    We know that it also sends $\MCommit(0, \id, c, D)$ at line~\ref{epaxos:cmt:send-cmt}.
    Moreover, by lines~\ref{epaxos:cmt:preaccok-precond} and \ref{epaxos:recvr:rec-precond}--\ref{epaxos:recvr:set-bal}, the coordinator sends $\MCommit$ before $\MRecoverOK$.
    Since all self-addressed messages are delivered immediately, it follows that $\vphase = \commitP$.
    However, because the condition at line~\ref{epaxos:recvr:is-cmt} does not hold, it must be that $\vphase \neq \commitP$: a contradiction.
  \item
    {\em Line~\ref{epaxos:recvr:break-cond-coord-end} (optimized protocol).}
    In this situation, the case at line~\ref{epaxos:recvr:break-cond-coord} is triggered.
    Since the conditions at lines~\ref{epaxos:recvr:coord-commit} and \ref{epaxos:recvr:coord-acc} do not hold, by Invariant~\ref{inv:bmax-acc}, process $p$ receives a $\MRecoverOK(b, \id, 0, c, \_, \_, \vphase)$ message from the coordinator $\coord(\id)$.
    However, before sending $\MRecoverOK$ the coordinator also sends $\MCommit(0, \id, c, D)$ at line~\ref{epaxos:cmt:send-cmt}.
    Since all self-addressed messages are delivered immediately, it follows that $\vphase = \commitP$, contradicting the fact that the condition at line~\ref{epaxos:recvr:coord-commit} does not hold.
  \item
    {\em Line~\ref{epaxos-thrifty:bad} (thrifty protocol).}
    In this situation, the case at line~\ref{epaxos:recvr:break-cond-coord}t is triggered.
    Since the conditions at lines~\ref{epaxos:recvr:coord-commit} and \ref{epaxos:recvr:coord-acc} do not hold, process $p$ receives either a $\MRecoverOK(b, \id, 0, c, \_, \_, \vphase)$ message from the coordinator $\coord(\id)$ or an $\MRecoverOK(b, \id, 0, c, D'', \_, \preacceptP)$ message from a process $q \neq \coord(\id)$ with $D'' \neq D$.
    The proof of the former case is identical to the proof above for the optimized protocol.
    For the latter case, since $q$ has pre-accepted $\id$, it belongs to the fast quorum of $\id$ (line~\ref{epaxos:cmt:send-preacc-thrifty}).
    But then, since $q \neq \coord(\id)$ and $D'' \neq D$, by line~\ref{epaxos:cmt:fast-cmt-thrifty}t and Invariant~\ref{inv:eq-dep}, the initial coordinator cannot take the fast path: a contradiction.
  \item
    {\em Lines~\ref{epaxos:recvr:send-acc-good}, \ref{epaxos:recvr:nop-2} (all protocol versions).}
    In this case $D'$ satisfies the condition at line~\ref{epaxos:recvr:fast-cond} (or line~\ref{epaxos-thrifty:fastcond}t in the thrifty version).
    Hence, $D' = D$ (Invariant~\ref{inv:eq-dep}) and $c' = c$ (Invariant~\ref{inv:cmd-or-noop}), as required.
  \item
    {\em Line~\ref{epaxos:recvr:conf-cmt-acc}, executed due to the first disjunct of line~\ref{epaxos:recvr:conf-cmt}, and line~\ref{epaxos:recvr:nop-1} (all protocol versions).}
    Before sending the $\MAccept$ message from one of these lines, $p$ has to satisfy the condition at line~\ref{epaxos:recvr:fast-cond} (\ref{epaxos-thrifty:fastcond}t).
    By Invariant~\ref{inv:eq-dep}, this condition can only be satisfied for the dependency set $D$.
    By lines~\ref{epaxos:recvr:conf-cmt},~\ref{epaxos:recvr:break-cond-cmt} and~\ref{epaxos:recvr:let-conflicts}, some process $q$ has committed an invalidating command $\id''$ with a payload $c'' \neq \noop$ and a dependency set $D''$, such that $\inconflict{c}{c''}$, $\id \notin D''$ and  $\id'' \notin D$.

    According to Proposition~\ref{prop:dep-from}, $D''$ is computed either from a quorum of processes $Q''$ that pre-accepted $\id''$, or during the recovery of $\id''$, after executing line~\ref{epaxos:recvr:send-acc-good} or line~\ref{epaxos:recvr:nop-2}.
    The former is impossible: in this case, since $Q_f \cap Q'' \neq \emptyset$, by Invariants~\ref{inv:cmd-or-noop} and~\ref{inv:two-preaccs} we would have $\id \in D''$ or $\id'' \in D$.

    Thus, some process $q''$ has executed line~\ref{epaxos:recvr:send-acc-good} or line~\ref{epaxos:recvr:nop-2} for $\id''$ and $D''$.
    Before that, $q''$ must have satisfied the condition at line~\ref{epaxos:recvr:fast-cond} (\ref{epaxos-thrifty:fastcond}t) with $D''$ after receiving $\MRecoverOK$ messages from a quorum $Q''$.
    Since $Q_f \cap Q'' \neq \emptyset$, there exists a process $q$ that pre-accepts $\id$ and sends a $\MValidateOK(\_, \id'', I'')$ message (line~\ref{epaxos:recvr:rec-validate-ok}).
    Notice that $q$ must have pre-accepted $\id$ before sending $\MValidateOK(\_, \id'', I'')$: otherwise, by line~\ref{epaxos:recvr:set-c}, Invariant~\ref{inv:cmd-or-noop} and the fact that $\noop$ conflicts with all commands, it would have included $\id''$ in its dependency set at line~\ref{epaxos:cmt:preacc-deps}, implying $\id'' \in D$.
    We now prove that $\id \in I''$, i.e., all three conjuncts at line~\ref{epaxos:recvr:let-conflicts} hold for $\id$ at $q$ when it computes $I''$:
    \begin{enumerate}[(i),leftmargin=21pt]
    \item $\id \notin D''$;  \label{epaxos-proof:lfp:1}
    \item $\phase[\id] = \commitP \Longrightarrow {\cmd[\id] \neq \noop} \land {\inconflict{\cmd[\id]}{c''}} \land {\id'' \notin \dep[\id]}$; and \label{epaxos-proof:lfp:2}
    \item $\phase[\id] \neq \commitP \Longrightarrow {\initCmd[\id] \neq \bot} \land {\inconflict{\initCmd[\id]}{c''}} \land {\id''\!\notin\initDep[\id]}$. \label{epaxos-proof:lfp:3}
    \end{enumerate}
    We have already established conjunct~\ref{epaxos-proof:lfp:1}.
    For conjunct~\ref{epaxos-proof:lfp:2}, by the induction hypothesis for case (b), $q$ must have committed $\id$ with command $c$ and dependency set $D$.
    We know $c \neq \noop$, $\inconflict{c}{c''}$, and $\id'' \notin D$, which implies the required.
    For conjunct~\ref{epaxos-proof:lfp:3}, recall that by the time $q$ sends $\MValidateOK(\_, \id'', I'')$ it has already pre-accepted $\id$.
    Thus, by lines~\ref{epaxos:cmd:set-initcmd}--\ref{epaxos:cmt:set-d0}, at that moment $q$ has $\initCmd[\id] = c$ and $\initDep[\id] = D_0$, where $D_0$ is the initial dependency set of $\id$.
    Since initial dependencies are always included in final dependencies, and $\id'' \notin D$, it follows that $\id'' \notin D_0$ as required.

    Now that we have proved that $\id \in I''$, it follows that $q''$ did not execute line~\ref{epaxos:recvr:send-acc-good}, and therefore executed line~\ref{epaxos:recvr:nop-2}.
    However, since $\id \in I''$, by line~\ref{epaxos:recvr:all-committed}, $\id$ is committed either with a payload $\noop$ or with dependencies that include $\id''$.
    Both cases contradict the induction hypothesis for case (b) and the fact that $\id$ went on the fast path.
  \item
    {\em Line~\ref{epaxos:recvr:nop-3}, executed due to the second disjunct of line~\ref{epaxos:recvr:conf-cmt} (optimized and thrifty protocols).}
    This case is impossible as it contradicts Lemma~\ref{prop:red-one}.
  \item
    {\em Line~\ref{epaxos:recvr:break-cond-wait-r-end} (baseline protocol).}
    By line~\ref{epaxos:recvr:break-cond-wait-r}, there exists a process $q \in Q$ that sends a $\MValidateOK(\_, \id, I)$ message, and there also exists a process that sends a $\MWaiting(\id')$ message, with $\id' \in I$.
    We know that when $q$ sends $\MValidateOK$, $\id'$ is not committed at $q$: otherwise, process $p$ would have satisfied the condition at line~\ref{epaxos:recvr:conf-cmt}, and thus would never have executed line~\ref{epaxos:recvr:break-cond-wait-r-end}.
    Thus, by line~\ref{epaxos:recvr:let-conflicts}, the initial payload of $\id'$ conflicts with $c$, and $\id$ is not in the initial dependency set of $\id'$.
    By the same line, we also have $\id' \notin D$.
    But then, by Lemma~\ref{prop:wait}, $\id$ did not take the fast path, contradicting our assumption.
    Therefore, this case is impossible.
  \item
    {\em Line~\ref{epaxos:recvr:break-cond-wait-r-end} (optimized and thrifty protocols).}
    We can show that these cases are impossible by applying the same reasoning as above, replacing Lemma~\ref{prop:wait} by Lemma~\ref{prop:wait-opt} for the optimized protocol or Lemma~\ref{prop:wait-opt-plus} for the thrifty protocol.
  \item
    {\em Line~\ref{epaxos:recvr:last-acc} (all protocol versions).}
    In this case the condition at line~\ref{epaxos:recvr:fast-cond} (\ref{epaxos-thrifty:fastcond}t) does not hold for $\id$.
    Thus, this case is impossible, as it contradicts Invariant~\ref{inv:eq-dep}.
  \end{itemize}

\emph{Case (b).}
  We make a case split on the line at which $p$ sends the $\MCommit$ message:
  \begin{itemize}
  \item \emph{Line~\ref{epaxos:cmt:send-cmt}.}
    The required follows from the fact that only the initial coordinator of $\id$ may send a $\MCommit$ message at this line.
  \item \emph{Line~\ref{epaxos:cmt:sent-cmt-slow}.}
  As the precondition in line~\ref{epaxos:cmt:receive-acceptok} holds, $p$ received a quorum of $\MAcceptOK(b',\id)$ messages.
  It follows that $p$ had previously sent an $\MAccept(b',\id,c',D')$ message.
  Then the required follows from the induction hypothesis for case (a).
  \item \emph{Line~\ref{epaxos:recvr:is-cmt} or~\ref{epaxos:recvr:coord-commit}.}
    The required follows immediately from the induction hypothesis for case (b).
\end{itemize}
\end{proof}

\begin{proof}[Proof of Invariant~\ref{inv:agr} (all protocol versions)]
  Suppose that two processes $p$ and $p'$ receive $\MCommit(b,\id,c,D)$ and $\MCommit(b',\id,c',D')$ messages, respectively.
  Assume first that $\id$ takes the fast path, i.e., at some point a process sends $\MCommit(0, \id, c'', D'')$ at line~\ref{epaxos:cmt:send-cmt}.
  Then, by Invariant~\ref{inv:cmt-cmt}, we have $c = c' = c''$ and $D = D' = D''$, as required.
  Assume now that $\id$ takes the slow path, i.e., no process ever sends $\MCommit(0, \id, c'', D'')$ at line~\ref{epaxos:cmt:send-cmt}.
  In this case it is easy to show by induction on the length of the execution that some quorum of processes have received $\MAccept(\_, \id, c, D)$ and replied with $\MAcceptOK(\_, \id)$.
  Then by Invariant~\ref{inv:agr-cmt}, $c = c'$ and $D = D'$, as required.
\end{proof}

\begin{proof}[Proof of Invariant~\ref{inv:cnst}  (all protocol versions)]
  Assume that messages $\MCommit(\_, \id, c, D)$ and $\MCommit(\_, \id', c',
  D')$ are sent, where $c \neq \noop$, $c' \neq \noop$ and $\inconflict{c}{c'}$.
  We prove that $\id \in D'$ or $\id' \in D$.
  Notice that both commands satisfy one of the two conditions of Proposition~\ref{prop:dep-from}.
  \begin{enumerate}
  \item
    Assume that both commands satisfy condition \ref{prop:dep-from-q} of Proposition~\ref{prop:dep-from}.
    Then there are two quorums $Q$ and $Q'$ such that every process $q \in Q$ (respectively, $q \in Q'$) has sent a $\MPreAcceptOK(\id,D_q[,\_])$ message (respectively, $\MPreAcceptOK(\id',D'_q[,\_])$) with $\bigcup_{q \in Q} D_q = D$ (respectively, $\bigcup_{q \in Q'} D'_q = D'$).
    Since $Q \cap Q' \neq \emptyset$, there exists a process that sends both $\MPreAcceptOK(\id,D_q[,\_])$ and $\MPreAcceptOK(\id',D'_q[,\_])$.
    By Invariant~\ref{inv:two-preaccs}, $\id \in D'_q \subseteq D'$ or $\id' \in D_q \subseteq D$, as required.
  \item
    \label{proof:inv:cnst:snd}
    Assume that $\id'$ satisfies condition \ref{prop:dep-from-q} of Proposition~\ref{prop:dep-from}, while $\id$ satisfies condition \ref{prop:dep-from-rec} (the symmetric case is analogous).
    Then there exists a quorum $Q'$ such that every process $q \in Q'$ has sent a $\MPreAcceptOK(\id',D'_q[,\_])$ message with $\bigcup_{q \in Q'} D'_q = D'$.
    Furthermore, a process $p$ has sent an $\MAccept(\_, \id, c, D)$ from line~\ref{epaxos:recvr:send-acc-good} or line~\ref{epaxos:recvr:nop-2}.
    Let $Q$ be a quorum that enables the corresponding invocation of the handler at line~\ref{epaxos:recvr:recoverok} for $\id$.
    Prior to executing line~\ref{epaxos:recvr:send-acc-good} or line~\ref{epaxos:recvr:nop-2}, process $p$ has sent $\MValidate(\_, \id, c, D[,\_])$ message and has received $\MValidateOK(\_, \id, I_q)$ from every process $q \in Q$ (lines~\ref{epaxos:recvr:send-validate}--\ref{epaxos:recvr:rec-validate-ok}).
    Take a process $q \in Q \cap Q'$; then this process sends $\MPreAcceptOK(\id',D'_q[,\_])$ and $\MValidateOK(\_, \id, I_q)$.
    If $q$ sends $\MValidateOK(\_, \id, I_q)$ first, then by line~\ref{epaxos:recvr:set-c} and Invariant~\ref{inv:cmd-or-noop}, at the moment $q$ sends $\MPreAcceptOK(\id',D'_q[,\_])$ it has $\cmd[\id] = c$ or $\cmd[\id] = \noop$.
    Since $c$ and $\noop$ both conflict with $c'$, by line~\ref{epaxos:cmt:preacc-deps}, $\id \in D'_q$, as required.

    Consider now the remaining case in which $q$ sends $\MPreAcceptOK(\id',D'_q[,\_])$ before $\MValidateOK(\_, \id, I_q)$.
    Then by lines~\ref{epaxos:cmd:set-initcmd}--\ref{epaxos:cmt:set-d0} we know that, when $q$ sends $\MValidateOK$, it has $\initCmd[\id'] = c' \bowtie c$ and $\initDep[\id'] = D'_0$, where $D'_0$ is the initial dependency set of $\id'$.
    We make a case split depending on whether $\id' \in I_q$ or not.

    Assume first that $\id' \notin I_q$: the command $\id'$ is not selected at line~\ref{epaxos:recvr:let-conflicts} upon validation of $\id$.
    Then one of the following must hold at process $q$:
    \begin{enumerate}[(i),leftmargin=22pt]
    \item $\id' \in D$; or \label{item:csi-1}
    \item $\phase[\id'] = \commitP$ and $\cmd[\id'] = \noop \lor \cmd[\id'] \not \bowtie c \lor \id \in \dep[\id']$; or \label{item:csi-2}
    \item $\phase[\id'] \neq \commitP$ and $\id \in \initDep[\id']$. \label{item:csi-3}
    \end{enumerate}
    In case~\ref{item:csi-1} the required property follows immediately.
    In case~\ref{item:csi-2}, Invariant~\ref{inv:agr} ensures that process $q$ commits $\id'$ with $\cmd[\id'] = c' \neq \noop$ and $\dep[\id'] = D'$.
    Since $c' \bowtie c$, this implies $\id \in D'$, as required.
    In case~\ref{item:csi-3}, recall that when $q$ sends the $\MValidateOK$ message it has
$\initDep[\id'] = D'_0 \subseteq D'$.
    Then $\id \in D'$, as required.

    Assume now that $\id' \in I_q$: the command $\id'$ potentially invalidates $\id$.
    Notice that since $I_q \neq \emptyset$, process $p$ could not have executed line~\ref{epaxos:recvr:send-acc-good}.
    Thus it sends $\MAccept(\_, \id, c, D)$ at line~\ref{epaxos:recvr:nop-2}.
    By line~\ref{epaxos:recvr:all-committed} it has committed $\id'$ either with a dependency that contains $\id$, or with $\noop$.
    However, by Invariant~\ref{inv:agr}, it has to be committed with $c' \neq \noop$ and $D'$.
    Hence, $\id \in D'$, as required.
  \item
    Assume that both commands satisfy condition \ref{prop:dep-from-rec} of Proposition~\ref{prop:dep-from}.
    In this case, they both validate their dependencies at their recovery quorums at lines~\ref{epaxos:recvr:send-validate}--\ref{epaxos:recvr:rec-validate-ok}.
    Thus, there exists a process that sends $\MValidateOK(\_, \id, I)$ and $\MValidateOK(\_, \id', I')$.
    Assume that it first sends $\MValidateOK(\_, \id', I')$ (the symmetric case is analogous).
    By line~\ref{epaxos:recvr:set-initcmd}, at the moment it sends $\MValidateOK(\_, \id, I)$ it has $\initCmd[\id'] = c' \neq \bot$.
    The proof concludes with the same case-splitting as before: either $\id'$ belongs to $I$, or it does not.
  \end{enumerate}
\end{proof}

\subsection{Proof of Consistency}
\labappendix{execution}

Having proved Invariants~\ref{inv:agr} and~\ref{inv:cnst} for all protocol versions, we now prove that the execution protocol in Figure~\ref{epaxos:exec} maintains the {\sf Consistency} property of the replicated state-machine.
To this end, we use two auxiliary lemmas.
These lemmas hold for any two committed commands $\id$ and $\id'$ with conflicting non-$\noop$ payloads $c$ and $c'$, respectively.

\begin{lemma}
  \labinv{exec:1}
  Assume that a process $p$ executes $c$ before executing $c'$, and if the process executes both commands, then this happens in different iterations of the loop at line~\ref{epaxos:exec:loop}.
  Then $\id' \notin \dep[\id]$.
\end{lemma}
\begin{proof}
  By contradiction, suppose $\id' \in \dep[\id]$.
  Let $G$ be the subgraph computed during the iteration of the loop when $\id$ is executed, so that $\id$ belongs to $G$ (line~\ref{epaxos:exec:let-c}).
  By the definition of $G$, when $p$ executes $c$, this process has committed $c$ and all its transitive dependencies.
  It follows that $\id'$ also belongs to $G$.
  By line~\ref{epaxos:exec:test}, by the end of this iteration of the loop, $p$ will execute $c'$.
  But then either $c'$ is executed in the same iteration of the loop as $c$, or it was executed before $c$: a contradiction.
\end{proof}

\begin{lemma}
  \labinv{exec:2}
  If $c$ is executed before $c'$ in the same iteration of the loop at line~\ref{epaxos:exec:loop}, then there is a path from $\id$ to $\id'$ in the subgraph $G$ computed at line~\ref{epaxos:exec:let-c}.
\end{lemma}
\begin{proof}
  By contradiction, assume that there is no a path from $\id$ to $\id'$ in $G$.
  This implies that $\id$ is not a transitive dependency of $\id'$.
  Then by Invariants~\ref{inv:agr} and~\ref{inv:cnst}, $\id'$ must be a dependency of $\id$.
  Thus, $\id$ and $\id'$ must be in different strongly connected components of $G$, and the one containing $\id'$ must be ordered before the one containing $\id$ at line~\ref{epaxos:exec:for-scc}.
  Then $c'$ is executed before $c$: a contradiction.
\end{proof}

\begin{proof}[Proof of Consistency]
The proof is done by contradiction.
Consider two commands with conflicting non-$\noop$ payloads $c$ and $c'$ and whose identifiers are $\id$ and $\id'$, respectively.
Suppose that a process $p$ executes $c$ before executing $c'$, while another process $q$ does the opposite.
There are two possible execution paths at $p$:
\begin{itemize}
\item If $p$ executes $c$ and $c'$, then this happens in different iterations of the loop at line~\ref{epaxos:exec:loop}.
By Lemma~\ref{inv:exec:1}, we get $\id' \notin \dep[\id]$.
Then by Invariants~\ref{inv:agr} and~\ref{inv:cnst}, we get $\id \in \dep[\id']$.
\item $p$ executes both $c$ and $c'$, and this happens in the same iteration of the loop at line~\ref{epaxos:exec:loop}.
Then by Lemma~\ref{inv:exec:2}, there is a dependency chain going from $\id$ to $\id'$.
\end{itemize}
Thus, in both cases, there is a dependency chain going from $\id$ to $\id'$.
Considering process $q$ and using the same arguments, we can show that there is a dependency chain from $\id'$ to $\id$. Thus, $\id$ and $\id'$ always belong to the same strongly connected component at line~\ref{epaxos:exec:for-scc}, so are always executed in the order of their identifiers: a contradiction.
\end{proof}

\subsection{Proof of Liveness}
\label{epaxos:appendix:liveness}

\begin{proof}[Proof of Lemma~\ref{prop:all-good} (all protocol versions)]
  The proof is done by contradiction: assume that some commands are never committed by some correct processes.
  Let $\mathit{Ids}$ be the set of all such commands.
  Let $t \geq \GST$ be a time such that the following properties hold:
  \begin{enumerate}[left=0pt, label={\Alph*)}, ref={\Alph*}]
  \item All failures have occurred before time $t$. \label{proof:time:1}
  \item For all $\id \in \mathit{Ids}$, after time $t$, $\id$ can be recovered only by the single (unique and correct) process $p_\id$, which is returned by $\Omega[\id]$ at all correct processes. \label{proof:time:2}
  \item For all $\id \in \mathit{Ids}$, after time $t$, process $p_\id$ recovers $\id$ infinitely many times. \label{proof:time:3}
  \item For all $\id \in \mathit{Ids}$, by time $t$, the recovery timers at all correct processses are sufficiently high (line~\ref{epaxos:rec-start:start}) for all correct processes to receive a $\MCommit$ message from $p_\id$ if this process does not wait at line~\ref{epaxos:recvr:wait}.
\label{proof:time:5}
  \item For all $\id \in \mathit{Ids}$, after time $t$, every recovery attempt of $\id$ by $p_\id$ blocks forever at line~\ref{epaxos:recvr:wait}.\label{proof:time:4}
  \item In the optimized and thrifty versions, for all $\id \in \mathit{Ids}$, no correct process ever commits $\id$. \label{proof:time:4.6}
  \item For all $\id \in \mathit{Ids}$, in every recovery attempt after time $t$, process $p_\id$ selects the same dependency set $D$ at line~\ref{epaxos:recvr:let-c} (\ref{epaxos-thrifty:let-d0}t), which we denote $D_\id$. \label{proof:time:7}
  \end{enumerate}
The justification for the existence of such a time $t$ is as follows:
\begin{itemize}
  \item
    Item~\ref{proof:time:1} is straightforward.
\item
  Item~\ref{proof:time:2} follows the definition of $\Omega[\id]$ and the check at line~\ref{epaxos:rec-start:rec}.
\item
  Item~\ref{proof:time:3} follows from lines~\ref{epaxos:rec-start:start}--\ref{epaxos:rec-start:check-commit} and the fact that each command in $\mathit{Ids}$ is never committed by at least one correct process.
\item
  Item~\ref{proof:time:5} holds because processes increase their recovery timers on each recovery attempt, and $p_\id$ increases its ballot number at each recovery attempt.
\item
  Item~\ref{proof:time:4} holds because of items~\ref{proof:time:1}--\ref{proof:time:5} and the fact that some commands are never committed by some correct processes.
\item
  Item~\ref{proof:time:4.6} holds because if a correct process ever commits $\id$, then by item~\ref{proof:time:5}, process $p_\id$ eventually executes line~\ref{epaxos:recvr:send-cmt} or \ref{epaxos:recvr:coord-commit}, contradicting item~\ref{proof:time:4}.
\item
  Item~\ref{proof:time:7} holds in non-thrifty versions because the set selected at line~\ref{epaxos:recvr:let-c} is the initial dependency set of $\id$, which is unique (line~\ref{epaxos:recvr:fast-cond}).
  In the thrifty version, by items~\ref{proof:time:1} and \ref{proof:time:5}, during any two recovery attempts of $\id \in \mathit{Ids}$ after time $t$, process $p_\id$ receives $\MRecoverOK$ messages from the same set of processes (lines~\ref{epaxos:recvr:recoverok} and \ref{epaxos-thrifty:bad-phase}t).
  If different subsets of this set were to pre-accept different dependency sets, then these recovery attempts would terminate at line~\ref{epaxos-thrifty:bad} or~\ref{epaxos-thrifty:last-nop}, contradicting item~\ref{proof:time:4}.
\end{itemize}
  
  Take any $\id \in \mathit{Ids}$.
  By item~\ref{proof:time:4} in each recovery attempt after $t$, $\id$ blocks on at least one potentially invalidating command (line~\ref{epaxos:recvr:wait}).
  Since $\id$ cannot block on committed commands (lines~\ref{epaxos:recvr:conf-cmt}, \ref{epaxos:recvr:break-cond-cmt}, \ref{epaxos:recvr:all-committed}), it follows that, after time $t$, $\id$ can eventually block only on commands in $\mathit{Ids}$.
  We can therefore further assume that the chosen time $t$ also satisfies:
  \begin{enumerate}[left=0pt, label={\Alph*)}, ref={\Alph*}]
    \setcounter{enumi}{7}
  \item For all $\id \in \mathit{Ids}$, every recovery attempt of $\id$ after time $t$ is potentially invalidated by some command in $\mathit{Ids}$. \label{proof:time:8}
  \end{enumerate}
  Since we consider only executions with finitely many submitted commands, there exists a command $\id' \in \mathit{Ids}$ (line~\ref{epaxos:recvr:let-conflicts}) on which $\id$ blocks infinitely often.
  Eventually, process $p_{\id'}$ broadcasts a $\MWaiting$ message for $\id'$ at line~\ref{epaxos:recvr:send-wait} (line~\ref{epaxos-thrifty:sendwait}t in the thrifty version), which $p_\id$ eventually receives.
  This concludes the proof for the baseline protocol, since in this case $p_\id$ unconditionally aborts the recovery at line~\ref{epaxos:recvr:break-cond-wait-r-end}, contradicting item~\ref{proof:time:4}.
  We now continue the proof for the optimized and thrifty versions.

  Since $\id$ repeatedly blocks on $\id'$ (and in particular never satisfies the first disjunct at line~\ref{epaxos:recvr:conf-cmt}), line~\ref{epaxos:recvr:let-conflicts} implies:
  \begin{enumerate}[(i),leftmargin=22pt]
  \item $\id' \notin D_\id$; \label{proof:lvns:i1}
  \item the initial payloads of $\id$ and $\id'$ conflict; and \label{proof:lvns:i2}
  \item $\id \notin D'_0$, where $D'_0$ is the initial dependency set of $\id'$. \label{proof:lvns:i3}
  \end{enumerate}
  Fix any recovery of $\id$ by $p_\id$ after time $t$ with a recovery quorum $Q$.
  By item~\ref{proof:time:4}, process $p_\id$ sends a $\MWaiting(\id, [\_,]k)$ message at line~\ref{epaxos:recvr:send-wait} (line~\ref{epaxos-thrifty:sendwait}t in the thrifty version).
  Fix any recovery of $\id'$ by $p_{\id'}$ after it has received $\MWaiting(\id, [\_,]k)$ from $p_\id$ and after all processes in $Q$ have received $\MValidate(\_, \id, \_, D_\id [,\_])$ from $p_\id$ (item~\ref{proof:time:5}).
  By line~\ref{epaxos:recvr:set-initcmd}, in the recovery quorum of $\id'$ there exists at least one correct process $q' \in Q$ that stores the initial payload of $\id$.
  We now make a case split depending on whether $\id \in D_{\id'}$ or not.
  \begin{enumerate}
  \item \label{proof:lvns:1} $\id \notin D_{\id'}$.
    By item~\ref{proof:time:4.6}, process $q'$ never commits $\id$.
    Then, $\id \notin D_{\id'}$ and items~\ref{proof:lvns:i1}--\ref{proof:lvns:i2} imply that, when $q'$ receives a $\MValidate(\_, \id', \_, D_{\id'}[,\_])$ message from the recovery by $p_{\id'}$, it includes $\id$ in the set of invalidating and potentially invalidating commands at line~\ref{epaxos:recvr:let-conflicts}.
    Recall that process $p_{\id'}$ has received $\MWaiting(\id, [\_,]k)$.
    Since by item~\ref{proof:time:4}, it does not satisfy the condition at line~\ref{epaxos:recvr:break-cond-wait-r}, we have $k \leq n - f - e$.
    Hence, when $p_\id$ recovers $\id$, it has $|\Rmax| \leq n - f - e$.
    On the other hand, by lines~\ref{epaxos:recvr:recoverok-precond} and~\ref{epaxos:recvr:let-rmax},
    $|\Rmax| \ge |Q| - e \ge n - f - e \ge |\Rmax|$, which implies $|\Rmax| = |Q| - e$.
    Since the conditions at lines~\ref{epaxos:recvr:coord-in} and~\ref{epaxos:recvr:break-cond-coord} (line~\ref{epaxos-thrifty:bad-phase}t in the thrifty version) never hold for $\id'$, by item~\ref{proof:time:5}, $\coord(\id')$ must be faulty.
    However, in this case, when $p_\id$ recovers $\id$, the condition at line~\ref{epaxos:recvr:break-cond-wait} must hold, allowing $p_\id$ to complete the recovery, which contradicts item~\ref{proof:time:4}.
  \item $\id \in D_{\id'}$.
    Consider first the non-thrifty version of the protocol.
    Then we have $D_{\id'} = D'_0$, and hence, $\id \in D_0'$, which contradicts \ref{proof:lvns:i3}.

    Consider now the thrifty version.
    There exists a sequence of commands $(\id_n)_{n\in\mathbb{N}}$ in $\mathit{Ids}$ (item~\ref{proof:time:8}), where each command is blocked by the next one, and such that $\id_1 = \id$ and $\id_2 = \id'$.
    Recall that we are considering a run with only finitely many submitted commands.
    Therefore, the sequence $(\id_n)_{n\in\mathbb{N}}$ must be circular.
    Notice that if there exists $n$, such that $\id_n \notin D_{\id_{n+1}}$ then the rest of the proof proceeds exactly as in the case where $\id \notin D_{\id'}$, which is proven above (case \ref{proof:lvns:1}).
    Hence, for all $n$, $\id_n \in D_{\id_{n+1}}$.

    Take any $i \in \mathbb{N}$.
    Since $\id_i$ blocks at line~\ref{epaxos:recvr:wait} infinitely often, by line~\ref{epaxos-thrifty:fastcond}t there exists at least one correct process $q$ (item~\ref{proof:time:1}) that has pre-accepted $\id_i$ with $D_i$.
    Since $\id_{i-1} \in D_{\id_i}$, process $q$ stores command $\id_{i-1}$ (line~\ref{epaxos:thrifty:pre-accept-deps}).
    By item~\ref{proof:time:5}, when process $p_{\id_{i-1}}$ recovers $\id_{i-1}$, it eventually receives a $\MRecoverOK(\_,\id_{i-1},\_,\_,\vdep_q,\_,\vphase_q)$ message from $q$.
    Since $q$ has pre-accepted $\id_i$ with $D_i$, and $\id_{i-1} \in D_i$, line~\ref{epaxos:thrifty:start-phase}t ensures that $\vphase_q \neq \startP$.
    Moreover, by item~\ref{proof:time:4}, $\vphase_q \neq \commitP$ and $\vphase_q \neq \acceptP$.
    This implies that $\vphase_q = \preacceptP$, and by lines~\ref{epaxos-thrifty:fastcond}t and \ref{epaxos-thrifty:bad-phase}t, together with item~\ref{proof:time:7}, we have $\vdep_q = D_{\id_{i-1}}$.
    This means that process $q$ has pre-accepted both $\id_{i-1}$ and $\id_i$, with dependencies $D_{\id_{i-1}}$ and $D_{\id_i}$, respectively.
    Since $\id_{i-2} \in D_{\id_{i-1}}$, process $q$ must also store $\id_{i-2}$.
    Applying the same reasoning, we conclude that $q$ has pre-accepted $\id_{i-2}$ with $D_{\id_{i-2}}$.
    By repeating this argument for all commands in the sequence $(\id_n)_{n\in\mathbb{N}}$, we conclude that process $q$ has pre-accepted each command $\id_n$ with dependencies $D_{\id_n}$.
    Notice that these dependencies form a cyclic dependency graph.
    However, processes cannot locally pre-accept cycles (lines~\ref{epaxos:cmt:send-preacc}t and \ref{epaxos:cmt:preacc-deps}), which leads to a contradiction.
  \end{enumerate}
\end{proof}

\begin{proof}[Proof of Liveness]
  As long as no process starts the recovery of a command, every submitted command is eventually delivered.
  If the run involves the recovery of a command $\id$ then, according to lines~\ref{epaxos:rec-start:check-commit}--\ref{epaxos:rec-start:send-try}, processes will continue attempting to recover $\id$ until it is committed.
  By Lemma~\ref{prop:all-good} we know that eventually there will be a time after which $\id$ is committed.
  Finally, by lines~\ref{epaxos:rec-start:check-nop}--\ref{epaxos:rec-start:submit}, if $\id$ is committed as a $\noop$, the correct process that originally submitted it will retry, eventually committing it with the initially submitted payload.
\end{proof}

\subsection{Proofs of Theorems~\ref{theo:upper-unopt}, \ref{theo:upper} and \ref{theo:upper-thrifty}}

We have already proven that all three versions of \epp satisfy {\sf Consistency} (\S\ref{appendix:execution}) and {\sf Liveness} (\S\ref{epaxos:appendix:liveness}) properties of the SMR specification (\S\ref{sec:smr}).
It is easy to see that they also satisfy {\sf Validity} and {\sf Integrity}.
Moreover, since {\sf Liveness} holds in every execution with at most $f$ crashes, all versions of \epp are $f$-resilient.
Finally, in \S\ref{sec:commit} we argued that the baseline \epp protocol is $e$-fast.
The same argument shows that the optimized protocol is $e$-fast, and the thrifty protocol is $0$-fast.

\section{A Liveness Bug in \epaxos Recovery}
\label{sec:bug}

\newcommand{\seq}{\mathsf{seq}}
\newcommand{\MPrepare}{{\tt Prepare}}
\newcommand{\MTentativePreAccept}{\texttt{TentativePreAccept}}
\newcommand{\MTentativePreAcceptOK}{\texttt{TentativePreAcceptOK}}

We consider a system of $n = 9$ processes $p_1, \ldots, p_9$, tolerating at most $f = 4$ failures.
The size of a fast quorum is $f + \lfloor \frac{f+1}{2} \rfloor = 6$.
Process $p_1$ is the initial coordinator of a command $c$ with an identifier $\id$, and process $p_2$ is the initial coordinator of a conflicting command $c'$ with an identifier $\id'$.
Following the protocol as specified in Moraru's PhD thesis \cite{epaxos-thesis}, the counter-example is as follows.
\begin{enumerate}
\item Process $p_1$ sends $\MPreAccept(c, \id, \emptyset, 0)$ to $\{p_1, p_5, p_6, p_7, p_8, p_9\}$.
\item Processes $p_8$ and $p_9$ pre-accept $\id$ with $\dep[\id] = \emptyset$ and $\seq[\id]=0$.
\item Process $p_8$ starts recovering $\id$.
  \begin{enumerate}
  \item It broadcasts $\MPrepare$ and receives $f+1 = 5$ replies from $\{p_2, p_3, p_4, p_8, p_9\}$ (page 41, step 1).
  \item Steps 2 on page 41 and 3--4 on page 42 are not applicable.
  \item However, there are $\lfloor \frac{f+1}{2} \rfloor = 2$ processes ($p_8$ and $p_9$) that pre-accepted $\id$ with $\dep[\id] = \emptyset$ and $\seq[\id]=0$ at ballot $0$ (the ``If'' part of step 5 on page 42).
    Hence, $p_8$ sends $\MTentativePreAccept(\id, c, \emptyset, 0)$ to $\{p_2, p_3, p_4\}$ (step 6 on page 42).
  \item Process $p_4$ receives $\MTentativePreAccept(\id, c, \emptyset, 0)$ from $p_8$ and pre-accepts $\id$ with $\dep[\id] = \emptyset$ and $\seq[\id]=0$.
  \end{enumerate}
\item Process $p_2$ sends $\MPreAccept(c', \id', \emptyset, 0)$ to $\{p_2, p_5, p_6, p_7, p_8, p_9\}$.
\item Processes $p_2$, $p_5$, $p_6$, and $p_7$ pre-accept $\id'$ with $\dep[\id'] = \emptyset$ and $\seq[\id']=0$.
\item Processes $p_5$, $p_6$, and $p_7$ pre-accept $\id$ with $\dep[\id] = \{\id'\}$ and $\seq[\id]=1$.
\item Process $p_2$ successfully commits $\id'$ with $\dep[\id'] = \{\id\}$ and $\seq[\id']=1$ on the slow path using a quorum $\{p_2, p_3, p_4, p_8, p_9\}$.
\item Processes $p_1$, $p_2$, $p_3$, and $p_9$ fail.
\item \label{bug:rec-id} Process $p_4$ starts recovering $\id$.
  \begin{enumerate}
  \item It broadcasts $\MPrepare$ and receives $f+1 = 5$ replies from $\{p_4, p_5, p_6, p_7, p_8\}$ (page 41, step 1).
  \item Steps 2 on page 41 and 3--4 on page 42 are not applicable.
  \item However, there are $3 > \lfloor \frac{f+1}{2} \rfloor$ processes ($p_5$, $p_6$, and $p_7$) that pre-accepted $\id$ with $\dep[\id] = \{\id'\}$ and $\seq[\id]=1$ at ballot $0$ (the ``If'' part of step 5 on page 42).
    But since all processes in $\{p_4, p_5, p_6, p_7, p_8\}$ have pre-accepted $\id$, steps 6 and 7 on page 42 have no effect.
    As a consequence, recovery does not make progress.
  \end{enumerate}
\item Process $p_4$ repeatedly attempts to recover $\id$ without success (i.e., it indefinitely repeats item~\ref{bug:rec-id}).
\end{enumerate}

\fi

\end{document}